\documentclass[11pt]{article}
\usepackage[utf8]{inputenc}
\usepackage{amsmath}
\usepackage{physics}
\usepackage{graphicx}
\usepackage{mathtools}
\usepackage{amssymb}
\usepackage{multicol}
\usepackage{empheq}
\usepackage{bm}

\usepackage{amsthm}
\usepackage{float}
\usepackage{hyperref}

\DeclareMathOperator{\poly}{\mathsf{poly}}
\DeclareMathOperator{\gen}{\mathsf{gen}}
\DeclareMathOperator{\itm}{\mathsf{item}}
\newcommand{\nlp}{\mathsf{NLP}}
\newcommand{\weight}{\mathsf{weight}}
\newcommand{\Left}{\mathsf{left}}
\newcommand{\Right}{\mathsf{right}}
\usepackage{tikz}
\newcommand{\opt}{\mathsf{OPT}}
\newcommand{\eps}{\epsilon}
\newcommand{\hold}{\mathsf{hold}}
\newcommand{\ord}{\mathsf{ord}}
\newcommand{\low}{\mathsf{low}}

\newcommand{\prev}{\mathsf{prev}}
\newcommand{\semitight}{\mathsf{semitight}}
\newcommand{\fracc}{\mathsf{frac}}

\DeclareMathOperator{\init}{\mathsf{init}}

\usepackage{algorithm}
\usepackage{algpseudocode}

\newcommand{\sol}{\mathsf{sol}}
\newcommand{\tight}{\mathsf{tight}}

\newcommand{\rjrpd}{\mathsf{rJRP}\text{-}\mathsf{D}}
\newcommand{\rjrp}{\mathsf{rJRP}}
\newcommand{\cjrpd}{\mathsf{CJRP}\text{-}\mathsf{D}}
\newcommand{\cjrp}{\mathsf{CJRP}}
\newcommand{\jrp}{\mathsf{JRP}}
\newcommand{\jrpd}{\mathsf{JRP}\text{-}\mathsf{D}}
\newcommand{\JRP}{\mathsf{JRP}}

\newtheorem{theorem}{Theorem}

\newtheorem{definition}[theorem]{Definition}
\newtheorem{lemma}[theorem]{Lemma}

\newtheorem{corollary}[theorem]{Corollary}

\newcommand{\lp}{\mathsf{LP}}
\newcommand{\LP}{\mathsf{LP}}

\tikzset{every picture/.style={line width=0.75pt}} 

\author{
  Varun Suriyanarayana\thanks{Cornell University}\\
  \texttt{vs478@cornell.edu}
  \and
  Varun Sivashankar\thanks{Microsoft Research, India}\\
  \texttt{varunsiva@ucla.edu}
  \and
  Siddharth Gollapudi\footnotemark[2]\\
  \texttt{sgollapu@berkeley.edu}
  \and
  David Shmoys\footnotemark[1]\\
  \texttt{david.shmoys@cornell.edu}
}
\title{Improved Approximation Algorithms\\ for the Joint Replenishment Problem with Outliers,\\ and with Fairness Constraints}

\begin{document}

\maketitle
\begin{abstract}
The {\it joint replenishment problem} ($\jrp$) is a classical inventory management problem. 
We consider a natural generalization \textit{with outliers}, 
where we are allowed to reject (that is, not service) a subset of demand points. In this paper, we are motivated by issues of {\it fairness} - if we do not serve all of the demands, we wish to ``spread out the pain'' in a balanced way among customers, communities, or any specified market segmentation.
One approach is to constrain the rejections allowed, and to have separate bounds for each given customer. In our most general setting, we consider a set of $C$ features, where each demand point has an associated rejection cost for each feature, and we have a given bound on the allowed rejection cost incurred in total for each feature. This generalizes a model of fairness introduced in earlier work on the Colorful $k-$Center problem in which (analogously) each demand point has a given color, and we bound the number of rejections of each color class. In the $\jrp$, we seek to balance the cost incurred by a fixed ordering overhead with the cost of maintaining on-hand inventory over a longer period in advance of when it is needed. More precisely, there a given set of item types, for which there is specified demand over a finite, discrete-time horizon, and placing any order at a given time incurs a general ordering cost and item-specific ordering costs (independent of the total demand serviced); in addition, for each unit of demand held in inventory for an interval of time, there is a corresponding item-specific holding cost incurred; the aim is to minimize the total cost.

We give the first constant approximation algorithms for the fairness-constrained $\jrp$ with a constant number of features; specifically, we give a $2.86$-approximation algorithm in this case. Even for the special case in which we bound the total (weighted) number of outliers, this performance guarantee improves upon bounds previously known for this case. Our approach is an LP-based algorithm that splits the instance into two subinstances. One is solved by a novel iterative rounding approach and the other by pipage-based rounding. The standard LP relaxation has an unbounded integrality gap, and hence another key element of our algorithm is to strengthen the relaxation by correctly guessing key attributes of the optimal solution, which are sufficiently concise, so that we can enumerate over all possible guesses in polynomial time - albeit exponential in $C$, the number of features.

\end{abstract}
\newpage
\section{Introduction}
In the well-studied \textit{Joint Replenishment Problem} ($\jrp$), we are given a discrete time horizon $[T] = \{1,\ldots,T\}$, a set of item types $[N] = \{1,\ldots,N\}$, and a set of demand points $D$. Each demand point $d = (i,t)$ has an item type $i \in [N]$ and a deadline $t \in [T]$. The goal is to generate a set of {\it replenishment orders} that serve these demands and minimize total cost. The cost of each replenishment order placed at a certain timestep is $K_0+\sum_{i \in S} K_i$, where $S$ is the set of item types included in the replenishment order. Note that the problem is uncapacitated: a single replenishment order incurs a cost of $K_i$ to serve any number of demand points of item type $i$. Additionally, each demand point has a holding cost $H^{it}_s$ that is a non-decreasing function of the lag between the time $s$ at which the corresponding order is placed and the demand point's deadline $t$. An important special case is \textit{Joint Replenishment Problem with Deadlines}, denoted $\jrpd$, in which $H^{it}_s \in \{0,\infty \}$ for all $(i,t) \in D, s \in [t]$. 

In practice, not all demands can be met by the supplier, and so a choice must be made as to which demands are served. Consequently, we wish to introduce notions of {\it fairness} to ensure that no group of customers is disproportionately affected by these unserved, rejected demands.
We propose a model in which there is a set of $C$ features; for each feature, each demand point has a feature-specific rejection cost and there is a given bound on the total allowed rejection cost that may be incurred. 
This framework is extremely powerful in capturing a range of settings: for example, this weight might be the cost of obtaining the demand by special order from another vendor. Our model of fairness generalizes one introduced by 
Bandyapadhyay et al.\cite{Bandyapadhyay0P19}
for the $k-$center problem in which (analogously) each demand point has a given color, and the number of rejections of each color class is bounded. In fact, for notational simplicity, we will present our results for (a weighted version of) this case; our techniques directly extend to the more general setting. 
More formally,
each demand $d$ has a given color $c \in [C]$; if it is not served, there is an associated non-negative weight $w_d^c$. The goal is to serve a subset of requests to minimize the classical $\jrp$ objective of total service cost, while ensuring that the total weight of demands not serviced is at most $R_c$ for each $c \in [C]$. We can also incorporate into the objective function a penalty $p_d$ for each demand $d$ rejected. We shall refer to this problem as $\cjrp$.
When there is one color and rejection penalties are all 0, this is called the {\it $\JRP$ with outliers}, denoted $\rjrp$. 

Addressing these fairness constraints require new approximation algorithm techniques that are sufficiently flexible to adapt to these requirements; we give the first constant approximation algorithms for $\cjrp$. Our algorithms build upon previous rounding techniques to first partition the input; one part is then addressed by a sophisticated use of pipage rounding \cite{AgeevS04}, whereas the other relies on iterative LP-rounding, leveraging structure gained by assuming an extreme point for each iteration.
Specifically, we give a deterministic $(2.86+\epsilon)$-approximation algorithm for $\cjrp$ with runtime
$\poly((NT)^{\frac{C^4}{\eps}})$; for $\rjrp$, earlier work of Chekuri et al.\cite{ChekuriQZ2019} implies a bound of 4.42, so even in this special case, our work yields a substantial improvement.
Although our results are LP-based, the natural relaxation has an unbounded integrality gap; we show that a constant amount of side information about the optimal solution yields  this improvement. 
This exponential dependence on $C$ is necessary, $\cjrp$ is provably a generalization of set cover with $C$ elements and therefore, if $P \neq NP$, we cannot obtain an approximation factor better than $O(\log C)$ in $\poly(NTC)$ time.


\begin{center}
\begin{tabular}{ |c|c|c|c|c|c|c|c| } 
 \hline
 Setting & $\jrpd$ & $\jrp$ & $\jrp$ with penalties & $\rjrpd$ & $\rjrp$ & $\cjrpd$ & $\cjrp$\\  \hline
 Prior work & 1.574 & 1.791 & 2.54 & 4.07 & 4.42 & $O(\log{C})$ * & $O(\log{C})$ *\\ \hline
 Our results & -- & -- & -- & $2.8+\epsilon$ & $2.86+\epsilon$ & $2.8+\epsilon$ & $2.86+\epsilon$\\ 
 \hline
\end{tabular}
\end{center}

Much is known about the classical $\jrp$. 
{Nonner and Souza}\cite{NonnerS09} prove that even $\jrpd$ is APX-hard and {Bienkowski et al.}\cite{BienkowskiBCDNSSY15} prove that
remains true when demand points have identical holding cost functions.
They also prove that the canonical LP relaxation for $\jrpd$ has an integrality gap of 1.245, but no such lower bound on the best approximation factor is known. From an approximation algorithm perspective, {Nonner and Souza}\cite{NonnerS09} give a pair of random shift LP rounding algorithms that combine to give a $5/3$-approximation for $\jrpd$. {Bienkowski et al.}\cite{BienkowskiBCDNSSY15} propose a better probability distribution for one of these, giving a $1.57$-approximation. With arbitrary holding costs, {Levi et al.}\cite{LeviRS06} provide a $1.8$-approximation which combines two random shift methods. {Bienkowski et al.}\cite{BienkowskiBCDNSSY15} improve this to $1.791$ by combining these methods with a third approach in which they randomly scale up the LP solution and then convert the instance into a $\jrpd$ instance which they solve using the aforementioned $1.57$-approximation.

Although $\rjrp$ never been studied explicitly, Charikar et al.\cite{CharikarKMN01}gave constant approximation algorithms for several outlier selection problems.
Krishnawamy et al.\cite{KrishnaswamyLS18} give a constant approximation for $k$-median with outliers and recently, $2$-approximation algorithms for $k$-center with outliers have been found (Chakraborty et al.\cite{ChakrabartyGK20} and Harris et al.\cite{HarrisPST19}). For any class of set cover problems that admit a $\beta$-approximation algorithm, Inamdar and Varadrajan\cite{InamdarV18} prove a $2\beta+2$-approximation for the outliers version and Chekuri et al.\cite{ChekuriQZ2019} improve this to $\frac{\beta + 1}{1-\frac{1}{e}}$, which would correspond to a $4.42$-approximation for $\rjrp$ and arbitrary holding costs and $4.07$-approximation for $\rjrpd$.

Fairness in algorithms has become an increasingly active area of research; for the colorful $k$-center problem,
Anegg et al.\cite{AneggAKZ20} showed that
 no constant approximation was possible under the exponential time hypothesis. Recently, Jia et al.\cite{JiaSS20} provided a 3-approximation with runtime exponential in $C^2$ and Anegg et al.\cite{AneggAKZ20} gave a 4-approximation with runtime exponential in $C$. Chekuri et al.\cite{ChekuriQZ2019} prove that for $C$ coverage constraints, they can provide a $O( \beta \log C)$-approximation but in $\poly(NTC)$.

As in all of the aforementioned problems with outliers, the canonical LP with extra rejection variables and color-wise rejection limits has an unbounded integrality gap. To overcome this, we will strengthen the LP relaxation by providing a limited amount of side information about the optimal solution. Our algorithm and its analysis will depend on a constant parameter $\epsilon$ that determines the amount of side information used in the formulation. This side information is chosen in a way that makes it feasible to enumerate over all possible values. In this case, we shall say that we ``guess'' the correct value. Specifically, we guess a large $\epsilon-$dependent constant number of timesteps with replenishment orders, the most expensive item-type orders and the most expensive holding costs associated with demand points.

This strengthened LP is then used to decompose the input into two subproblems (which we call Instance 1 and Instance 2). The key element is that the fractional solution can be used to extract a collection of order points at which we commit to placing an order (and use the LP to bound their cost), but these serve only a subset of the demands points (even fractionally); this is Instance 1. The remainder of the demand points constitute Instance 2. The resulting structure of these two instances is different, and we exploit this in designing the separate approximation algorithms for them; the final solution combines the two solutions obtained. This approach also obtains the same approximation factor if the fixed ordering costs depend on time provided we guess some of the most expensive fixed orders the optimal solution will place. 

\section{Problem Formulation}
In this section, we first formally describe the joint replenishment problem and its robust and colorful variants. We then present an intuition-building sample instance. 

\paragraph{JRP:} In the Joint Replenishment Problem, we are given a set of item types $[N] = 1,\ldots,N$, timesteps $[T] = 1,\ldots,T$, a general ordering cost $K_0 \geq 0$ and item ordering costs $K_i \geq 0$ for each item type $i$. We are also given a set of demand points $D$. Each demand point has an item type and a deadline. We will often refer to demand points as $d \in D$ or $(i,t) \in D$ where $i$ is the item type and $t$ is the deadline. Every demand point must be ``serviced" by its deadline. However, each demand point $(i,t)$ also has a holding cost function $H^{it}_s$ which is the holding cost associated with servicing $(i,t)$ at time $s \leq t$. The holding cost function is monotone and decreases as $s$ increases (i.e., the holding cost becomes smaller the closer to deadline the demand point is serviced. Additionally, \textbf{Joint Replenishment Problem with Deadlines}, denoted $\jrpd$ is the term used to describe the special case in which $H^{it}_s \in \{0,\infty \}$ for all $(i,t) \in D, s \in [t]$

Demand points are serviced by ``replenishment orders". Every replenishment order $O$ is placed at some timestep $s$. The holding costs incurred by any demand point $(i,t)$ served by $O$ is $H^{it}_s$ The replenishment order $O$ incurs two types of cost: a fixed order opening cost $K_0$ and extra (item ordering) costs $K_i$ for every item type $i$ such that a demand point of type $i$ is serviced by order $O$. Note that replenishment orders are uncapacitated in the sense that the cost incurred is the same whether only a single demand of item type $i$ is served or many demands points of item type $i$ are served by the same order. 

Observe that there is always at most one replenishment order in any given timestep (if there were multiple orders at a timestep, they could be merged into a single order, reducing the total fixed order opening cost incurred but not increasing the item ordering costs or the holding costs.)

Since every demand point must be serviced, if there are multiple demands of the same item type with the same deadline, they can be treated as a single demand point and the holding cost of this merged demand point would be the sum of the holding costs of its components.

\paragraph{Robust JRP:} One of the problems we consider is $\rjrp$, where each demand point $d$ will have a weight $w_d$ and we will have a rejection weight limit $R$. In this variant, we are no longer required to serve every demand point (we will think of these demand points as being ``rejected"). However, the total weight of rejected demand points is at most $R$.

The idea of merging demand points breaks down when we are allowed to reject some demand points (i.e. we are allowed to not service some demand points at all). Nevertheless, if there is a timestep $t$ with $k$ demand points of item type $i$ having deadline $t$, we can simply split the timestep $t$ into $k$ timesteps $t_1,t_2,\ldots,t_k$, each being the deadline of one of the $k$ demand points of item type $i$. All demands of other types with deadline $t$ will have deadline $t_k$. The holding cost associated with serving any demand point at $t_j$ will be the same as the holding cost of serving it at $t$. This increases the length of our time horizon by at most a factor of $|D|$, which is clearly polynomial in the size of the input instance. This concept of splitting timesteps will reincarnate itself in a different manner when setting up one of our rounding algorithms. 

\paragraph{Assumption:} The above discussion on splitting timesteps allows us to, without loss of generality, simplify notation (and make notation consistent with past work on Joint Replenishment), assume that for every $(i,t) \in [N] \times [T]$, there is at most one demand point of item type $i$ with deadline $t$. Note that several $(i,t)$ pairs may not correspond to any demand point. 

\paragraph{Colorful JRP:} We will also generalize our ideas to what we call $\cjrp$, where instead of having a single rejection limit, we have multiple colors $[C] = 1,\ldots,C$. In this variant, each demand point will also have a color. We do not assume any correlation between the item type and the color of a demand point. Instead of a single rejection weight limit as in $\rjrp$, we now have rejection weight limits $R_c$ for each color $c$. We will think of $w^c_{d}$ as being $0$ if $d \in D$ is not of color $c$. If $d$ is of color $c$ then $w^c_{d}$ will be the weight of the demand point. It will turn out that the assumption that each demand point has weight in only a single color is unnecessary. Furthermore, in $\cjrp$, we will allow each demand point $d$ to also have a rejection penalty $p_d$. In addition to the ordering and holding costs, we will also have a rejection penalty cost which is the sum of the penalties corresponding to rejected demand points.

\subsubsection*{ILP Formulation}

We think of a replenishment order as having 2 parts: a general order and several item orders. In the formulation below, $y_s \in \{0,1\}$ indicates whether a replenishment order at time $s$ is placed (the \textit{general order}), $y^i_s \in \{0,1\}$ indicate whether the order at time $s$ includes item $i$ (an \textit{item order} of item $i$) and $x^{it}_s \in \{0,1\}$ indicate whether demand $(i,t)$ is served at time $s$ and let $r_{it} \in \{0,1\}$ indicate whether the demand $(i,t)$ is rejected. In the LP relaxation, the constraints that variables must be $\{0,1\}$ are replaced with non-negativity constraints. None of the variables need ever exceed $1$ in an optimal LP solution.

\begin{empheq}{align*}
\text{minimize}\quad & \sum_{s=1}^T y_s K_0 + \sum_{i=1}^N \sum_{s = 1}^T y^i_s K_i + \sum_{(i,t) \in D} \sum_{s=1}^T H_s^{it} x_s^{it} + \sum_{(i,t) \in D} p_{it}r_{it}\\
\text{subject to } & r_{it} + \sum_{s = 1}^t x_s^{it} \geq 1 \quad \text{for each } (i,t) \in D; &&\hspace{-0.8in}\text{{(Reject or service every demand)}}\\
& y^i_s \leq y_s \quad \text{for each } i \in [N], s \in [T]; &&\hspace{-0.8in} \text{{(Item order requires a general order)}}\\
& x_s^{it} \leq y_s^i \quad \text{for each } (i,t) \in D, s \in [T]; &&\hspace{-0.8in}\text{{(Servicing requires an item order)}}\\
& \sum_{(i,t)\in D} w^c_{it} r_{it} \leq R_c \quad \text{for each } c \in [C]; &&\hspace{-0.8in} \text{{(Rejection bound for each color)}}\\
&y_s,y_s^i,x^{it}_s,r_{it} \in \{0,1\} \quad \text{for each } s \in [T], i \in [N], (i,t) \in D 
\end{empheq}

In the LP relaxation, the constraints that variables must be assigned values in $\{0,1\}$ needs to be replaced with constraints ensuring that the variables are non-negative. None of the variables will ever exceed $1$ in an optimal solution. 

In a slight abuse of notation, we will often use $LP^{sol}$ to refer to the solution and its cost. Which meaning is applicable will be clear from context. Furthermore, we will use $LP^{sol}_{\gen}$ to refer to the general ordering cost of the solution $(\sum_{s=1}^T y_sK_0)$ and $LP^{sol}_{\itm}$ to refer to its item ordering cost $(\sum_{s=1}^T\sum_{i=1}^N y^i_s K_i)$ 

\subsubsection*{Intuition}
We now discuss the solution to a small instance of $\cjrpd$ to build some intuition for the problem.
\begin{figure}[h]
\includegraphics[width=12cm]{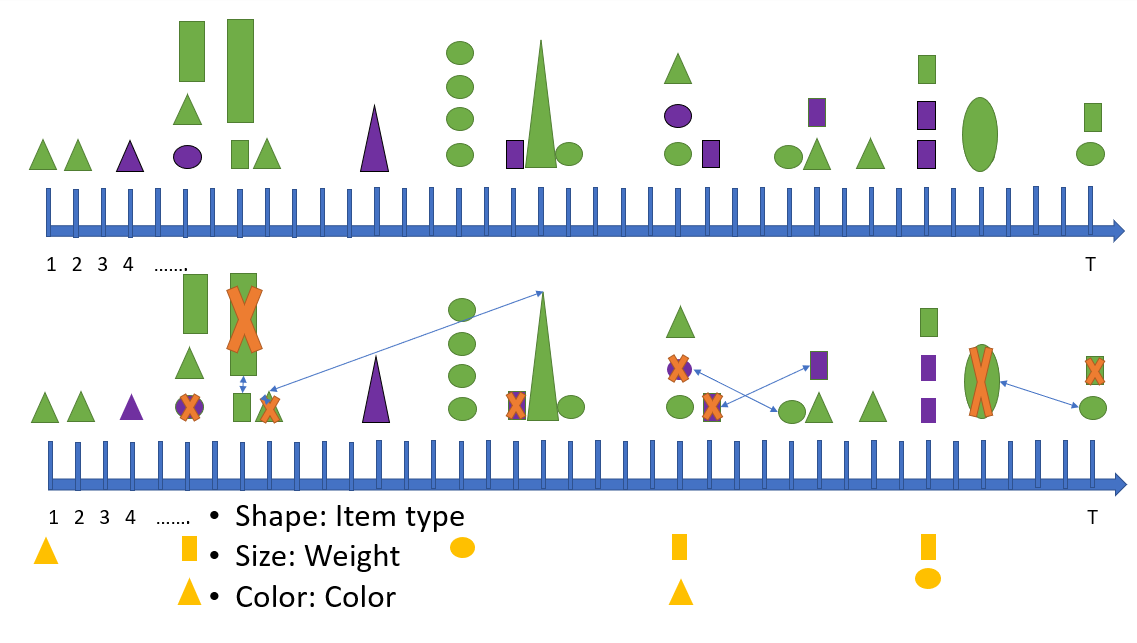}
\centering
\end{figure}

The figure above shows an example instance with 2 colors and 3 item types. The top half of the figure is the input instance, the bottom half is the returned solution, with rejections corresponding to red $X$ and the golden shapes corresponding to items included in a replenishment order. For example, at time 1 there is an order serving only triangles and the next order is at time 6 serving demands of both rectangle and triangle types. Every demand point that isn't rejected is serviced at the last possible replenishment order including that items' type.

At a superficial glance, there are some decisions that the optimal solutions makes (the pairs of demand points with thin blue lines connecting them) which may not seem intuitive. Possible explanations are:

\begin{enumerate}
    \item Why is the smaller triangle at timestep 9 rejected but the larger triangle with a much later deadline serviced by an order at timestep 6 serviced?

    The larger triangle has a much higher weight, so rejecting it uses up too much of the total weight of green rejections allowed.
    
    \item At timestep 8, why is the larger rectangle rejected but the smaller one served?

    The smaller rectangle has a far higher rejection penalty than the larger one, making it more expensive to reject.
    
    \item Why is the second last green circle rejected but the last green circle (which has a later deadline and a smaller weight) served?

    Holding costs are entirely demand point dependent. In this case, the second last green circles' holding cost was much higher than the little green demand points' holding cost at the time the last replenishment order is placed.

    This is the same reason the purple rectangle at around $2T/3$ is rejected but the purple rectangle after that is accepted.

    \item Why is the purple circle with deadline at the same time as the second last replenishment order rejected but the green one with a later deadline serviced?

    The green rejection budget has been fully used by other demand points that were rejected.
\end{enumerate}

\paragraph{A simple 2-approximation for $\jrpd$}: Solve a simplified LP in which we replace $x^{it}_s$ with $y^i_s$ if $H^{it}_s =0$, and with $0$ if $H^{it}_s=\infty$. Let $Z_t = \sum_{s \leq t} y_s$. We place a general order at the earliest timestep such that $Z_t$ is at least  $k$ for each $k \in \mathbb{N}$. For each item type, we also define $Z^i_t=\sum_{s \leq t} y^i_s$ and place tentative item orders at the earliest timestep such that $Z^i_t$ is at least $k$ for each  $k \in \mathbb{N}$. For each tentative item order, place an item order to coincide with the previous and the next general order. It is not hard to see that this algorithm yields a feasible solution of cost at most $\LP^{\sol}_{\gen}+2\LP^{\sol}_{\itm}$. Thus, given these general orders, we have $N$ decoupled single-item lot-sizing inputs, which can also be solved by dynamic programming\cite{WagnerW58}. The best known algorithms for $\jrpd$ extend this idea much further by randomizing the placement of general orders and making their spacing non-uniform. 

This basic idea is one component of our algorithm but is insufficient for $\cjrpd$; the approach can be modified to ensure that any single color's rejection limit is met in expectation, but this does not imply that we can find a solution of low cost that simultaneously satisfies each color's rejection limit. To achieve our goal, we will use the initial general orders from above to split the input instance into two smaller instances, each with special structure. We can solve each of these instances and combine the solutions to obtain $(3+\epsilon)$-approximation for $\cjrpd$. Some careful refinement allows us to improve this to $2.8+\epsilon$. The ideas that drive these improvements help us extend our algorithms to more general holding cost functions and achieve an almost as good bound of $2.86+\epsilon$.

The structure of $\jrp$ and the variants we study suggests that dynamic programming might play an important role here. Indeed, when either $K_0=0$ or there is only a single item type and $w^c_d \in \{0,1\}$ then $\rjrp$ can be solved exactly by dynamic programming. However, the utility of dynamic programming without these assumptions is not clear; the best approximation algorithms do not use dynamic programming at all and instead attempt to round an LP solution.

The two smaller instances our algorithm will construct will be solvable by dynamic programming, like the placement of item orders in the aforementioned $2-$approximation. However, because these instances were constructed depending on $\LP^{\sol}$ and the optimal solution to the original problem might not be splittable in the same manner, we need algorithms with guarantees in terms of $\LP^{\sol}$ for these instances, which the dynamic program doesn't provide. To overcome this, we devise two rounding algorithms, one for each instance. However, as we will see below, the rounding algorithms will only give the desired guarantees if the original LP solution incorporates some additional side information.

For the same reasons as in facility location with outliers\cite{CharikarKMN01}and other clustering problems with outliers (\cite{Bandyapadhyay0P19}, \cite{Chen08}), this LP relaxation has an unbounded integrality gap even if there is only a single item type and holding costs are $\{0,\infty\}$.

\begin{lemma}
The integrality gap of the formulation above is unbounded.
\end{lemma}

\begin{proof}
Consider an instance with a single item type, $K_0=1, K_1=0, H^{it}_s=0$ for all $ s \leq t$. For each $t \in [T]$, let there be a demand point with deadline $t$. Let $R=T-1$. The optimal solution must open an item order and serve at least 1 demand point, so the optimal solution has cost 1. However, the solution $r_{it}=\frac{T-1}{T},y^1_1=y_1=\frac{1}{T}$ is clearly feasible and costs $\frac{1}{T}$. This implies that the integrality gap is $T$, which can be arbitrarily large.
\end{proof}

However, just like in those settings, there is a simple remedy; guess a $C-$dependent constant amount of side information concerning the most expensive components of the optimal solution and augment the LP relaxation with this information. The following theorem which is proven in the appendix suggests that the amount of side information must depend on $C$ if we hope to obtain a constant approximation. The following lemma is proven using the same idea as the aforementioned 2-approximation for $\jrpd$ with a touch of randomization.




\begin{lemma}
Suppose the optimal LP solution to the above formulation with one item type has $\sum_{s} y^i_s = Q$ (and consequently costs $Q(K_0+K_1)$, then there exists an integral solution with cost $\lceil Q \rceil$ (and costs $\lceil Q\rceil (K_0+K_1)$
\label{random_shift_lemma}
\end{lemma}

\begin{proof}
Let $Z_t=\sum_{s=1}^t y^i_s=\sum_{s=1}^t y_s$. (By our constraints, $y_s \geq y^i_s$. If $y_s > y^i_s$ for any timestep, then we can decrease $y_s$ to $y^i_s$ and still preserve feasibility.)
Uniformly at random, choose $\lambda \in (0,1)$. For all $k \in 0,\ldots,\lfloor Z_t \rfloor$, we open orders of item $i$ at the timestep in which $Z_t$ first exceeds $k+\lambda$. Depending on our choice of $\lambda$, there are either $\lfloor Q \rfloor$ or $\lceil Q \rceil$ many orders opened this way.

It can be argued that the probability a demand point with 0-holding cost interval $I$ is satisfied is $\min(\sum_{s \in I} y^i_s,1)$ (details in full version of the paper). Therefore, in expectation, the number of demand points rejected is $\sum_{d \in D} r_d \leq R$.

We know that with probability 1, the solution costs at most $\lceil Q \rceil$ and that with positive probability, it rejects at most $R$ demand points giving us the desired result.
\end{proof}

The following lemma shows that set cover with $C$ elements reduces to $\cjrpd$, implying that we cannot achieve a constant factor approximation with runtime $\poly(DTC)$ unless $P = NP$. This demonstrates the complexity of dealing with multiple colors.
\begin{lemma}
Given an instance of set cover with $C$ elements, we can construct an equivalent $\cjrp$ instance with 1 item type, $C$ colors and $0-\infty$ holding costs.
\end{lemma}

\begin{proof}
All demand points will incur holding cost 0 if they are served at their deadline and $\infty$ otherwise. Number the sets in the set cover instance $1, \ldots, T$. Number the elements of the set cover instance $1, \ldots, C$. For every element $c$ of set $t$, there exists a demand point of deadline $t$ and color $c$. $K_0=1$ and $K_i=0$ for the single item type. For every color, we are allowed to reject all but one demand of that color. Observe that every solution to the set cover problem corresponds to a set of timesteps in the $\jrp$ problem and placing replenishment orders in these timesteps leads to cost equal to the number of sets in the cover. Since every element $c$ has a set $t$ in the cover that contains it, the replenishment order at time $t$ serves a demand of color $c$. Similarly, every solution to the $\jrp$ problem corresponds to a set of timesteps and the collection of sets corresponding to these is a valid cover if the original solution is feasible.
\end{proof}

\section{$\rjrpd$: Deadlines and a Single Color}

We will present an algorithm for $\rjrpd$ that yields the following theorem. 


\begin{theorem}
\label{thm:jrpd_bound}
Given an instance of $\rjrpd$, then for any $\epsilon > 0$, we can obtain a solution of cost $(2+\epsilon)\LP^{\sol}_{\gen}+(3+\epsilon)\LP^{\sol}_{\itm} \leq (3+\epsilon) \opt$ with runtime $\poly((NT)^{\frac{1}{\eps}})$.
\end{theorem}


The algorithm starts 
by augmenting the LP relaxation with side information about the optimal solution -- sufficiently compact to permit enumeration over all possible values parameterized by an integer $M$ dependent on $\eps$; the choice of $\eps$ determines both $M$ and the running time of the algorithm, with an obvious trade-off between them. 
Then, inspired by the $2-$approximation for $\jrpd$, we use the LP solution to identify a set of initial general orders; as in that algorithm, actually solve an adapted LP for $\rjrpd$ where we replace $x^{it}_s$ with $y^i_s$ if $H^{it}_s =0$, and with $0$ if $H^{it}_s=\infty$. We split the problem into two instances, where the 
the first instance consists of those timesteps at which we have placed general orders and contains all of the demand points  that can be serviced by these orders. There is a feasible LP solution to this instance of $\rjrpd$ in which the item ordering cost is at most $2\LP^{\sol}_{\itm}$, in the same way that the item ordering costs doubled in the simple $\jrpd$ algorithm above. The second instance consists of all timesteps in which we have not placed an initial general order, and the remaining demand points. $\LP^{\sol}$ restricted to the timesteps and demands in instance 2 is feasible for instance 2. We construct an LP-rounding $(1+\epsilon)$-approximation algorithms for both instances and combine the solutions to obtain a $(3+\epsilon)$-approximate solution.

We will use $I(i,t)$ and/or $I_{it}$ to denote the interval for which the holding cost corresponding to $(i,t)$ is 0. When there are sufficiently many orders in the optimal solution, we will be able to amortize ``rounding costs" as being negligible. Let $M \geq 3$ be sufficiently large to ensure that $\frac{10 \log M}{M} \leq \frac{\eps}{6}$: choosing $M = 1000 + \frac{1000}{\eps}\log(\frac{1000}{\eps})$ suffices.  
We will assume that the optimal solution has at least $M$ replenishment orders. (When the optimal solution has fewer than $M$ replenishment orders, by enumeration we guess the corresponding timesteps $\mathcal{T}_{small}$ and generate item orders by solving the LP restricted to $\mathcal{T}_{small}$ and with the constraints $y_s=1$ for all $s \in \mathcal{T}_{small}$ and applying the iterative rounding scheme in Section 4.2.) The algorithm proceeds as follows, where we set $\epsilon = \epsilon (M)$:
\begin{itemize}
    \item[(a)] \textbf{Enumeration for item orders} We guess the $M$ most expensive item orders in the optimal solution. Add constraints setting the corresponding $y^i_s=1$. Let $K^{M}_{\max}$ be the cost of the cheapest guessed item order. For all $(i,s)$ not in the guessed set such that $K_i > K^{M}_{\max}$, add the constraint $y^i_s=0$. Solve the LP (which must have a feasible solution if our guess is correct). We will need to round several item ordering variables up to $1$ and knowing that every fractional item ordering variable has cost at most $K^{M}_{\max}$ is critical in bounding the cost.
    \item[(b)] \textbf{Initial General Orders (IGO)} Let $Z_t=\sum_{s \leq t} y_s$. Let $\mathcal{T}$ be the set of timesteps $t'$ in which the running sum $Z_{t'}$ first reaches $0,1,2,3,\ldots, \lfloor Z_T \rfloor, Z_T$. Place a general order at each time $t \in \mathcal{T}$. Using these timesteps, we split the instance into two and give a $(1+\epsilon)$-approximation for each.
    \item[(c)] \textbf{Instance 1} This consists of timesteps $\mathcal{T}$ and the set of demand points $(i,t)$ such that there exists $s \in \mathcal{T}$ for which $H_s^{it} = 0$; let $D_1$ denote this set. The rejection limit for this instance is at most the extent to which the LP solution rejected demands in $D_1$: $\sum_{(i,t) \in D_1} r_{it}$. We solve this instance via iterative rounding by treating $K_0=0$ for each $t \in \mathcal{T}$, since the general orders are already opened. This iterative rounding procedure is critical to solving the multi-colored version of the problem. The key structural feature of this instance is that because there is a general order at each timestep $t \in \mathcal{T}$, the item types are (somewhat) decoupled. The only thing connecting them is the shared rejection limit. This enables us to obtain the following guarantee:
\end{itemize}
    \begin{theorem}
    \label{thm:iterative_special}
    Given an LP solution to Instance 1, we can obtain an
    integral solution of cost not more than $(1+\frac{\epsilon}{4})$ times the LP solution cost. 
    \end{theorem}
\begin{itemize}
\item[(d)] \textbf{Instance 2} This consists of timesteps $[T]-\mathcal{T}$ and all demand points $D_2$ that are not in $D_1$. Note that unlike vanilla $\jrp$, because we are allowed to reject demand points, plenty of partially served demands might have no IGO in their 0-holding cost interval. Therefore, we need an Instance 2 to accomodate these. The rejection limit for this instance is $\sum_{(i,t) \in D_2} r_{it}$. We will add the constraint that between any two consecutive initial general orders $s_1$ and $s_2$, $\sum_{s=s_1+1}^{s_2-1} y_s \leq 1$.  The pipage rounding framework relies on being able to partition the time horizon, non-integral replenishment orders, and, critically, demand points into ``batches'', each of which has $\sum_{s} y_s \leq 1$. Maintaining this invariant ensures that the rejection variables are a linear function (instead of piecewise linear) of the $y$ variables which is essential in obtaining the following guarantee:
\end{itemize}
    \begin{theorem}
    \label{thm: pipage_special}
    Given an LP solution to Instance 2, we can obtain an
    integral solution of cost not more than $(1+\frac{\epsilon}{6})$   times the LP cost. 
    \end{theorem}
\begin{itemize}
    \item[(e)] \textbf{Forming the solution} We simply merge the solutions from the two instances; if a replenishment order exists in either solution, it will be in our final answer.
\end{itemize}

In the remainder of this section we give details of the algorithms to round LP solutions Instance 2 followed by Instance 1. (The order is swapped for ease of exposition). The following analysis will also operate on the assumption that we have correctly guessed the $M$ most expensive item orders. In this case, we do obtain a feasible LP solution that costs at most $\opt$ (because the optimal integral solution is feasible). Also notice that the number of initial general orders placed is at most $Z_T+2$, and the cost of these general orders is at most $(1+\epsilon/8)\LP^{\sol}_{\gen}$. Moreover, the item ordering cost of the optimal LP solution for Instance 1 is at most $2\LP^{\sol}_{\itm}$ for the same reason as for the $2-$approximation to $\jrpd$. The integral solution to Instance 2 that we obtain costs at most $(1+\epsilon/6) \opt$. Putting all this together, we obtain Theorem~\ref{thm:jrpd_bound}
\subsection{$\rjrpd$: Solving Instance 2 with Pipage Rounding}
We devise a novel pipage-style rounding method to round a solution $(y,r)$ of cost $\LP^{\sol}$ to Instance 2. We obtain an integral solution of cost at most $(1+\frac{\epsilon}{6})\LP^{\sol}$. Our algorithm will consider groups of non-integral variables and adjust until at least one of them is integral, while maintaining feasibility and not increasing cost. At the end, we will need to round up at most 2 general orders and 2 item orders, which is why our solution might cost slightly more than $\LP^{\sol}$. We believe that this is the first use of pipage rounding in an inventory model setting. The algorithm described below can be optimized to need slightly fewer guesses, but this approach generalizes more smoothly to $\cjrp$.

This entire section will assume that the entire time horizon in the input instance to the pipage rounding algorithm (Instance 2 from the previous subsection) can be partitioned into a collection of intervals called ``batches", that for all batches $B,$ $\sum_{s \in B} y_s \leq 1$ and that for any demand point $(i,t)$ with $0-$holding cost interval $I_{(i,t)}$, there exists a batch $B(i,t)$ such that $I_{(i,t)} \subset B(i,t)$.
First, we make a critical structural observation: $\sum_{s \in I_{i,t}} y^i_s \leq \sum_{s \in I_{i,t}} y_s \leq 1$ and therefore, we can assume that $r_{it} = 1 - \sum_{s \in I_{i,t}} y^i_s$. (Note that in the original LP solution, there may be demands with extremely large intervals such that $\sum_{s \in I_{(i,t)}} y^i_s > 1$ and $r_{it}=0$, however, instance 2 cannot have such demand points because every interval $I$ with $\sum_{s \in I} y_s \geq 1$ must have an IGO and any demand point with $I_{(i,t)} = I$ would be in instance 1).

This implies the important fact below

\textbf{Fact:} $\sum_{d \in D} r_d = 1 - \sum_{s} y^i_s n^i_s$ where $n^i_s$ is the number of demand points $(i,t)$ such that $s \in I(i,t)$.

\begin{itemize}
\item[1.] \textbf{Splitting Timesteps}  We first show that WLOG, $y^i_s \in \{0, y_s\}$ for each timestep $s \not \in \mathcal{T}$
by splitting these timesteps as follows: sort the items by their $y^i_s$ value; WLOG assume that $y^{1}_s \leq y^{2}_{s} \leq \dots \leq  y^{N}_s$. If $y^1_s = y^N_s$, then WLOG, $ y^i_s =y_s$ for each $i \in [N]$.
Consider an index $i$ for which $y^i_s < y^{i+1}_s$. We shall split timestep $s$ into two timesteps $s_1$ and $s_2$ with $y_s = y_{s_1} + y_{s_2}$ and
$y_{s_1} = y^i_s$; then set $y^{j}_{s_1} = y^{j}_{s}$ and
$y^{j}_{s_2} = 0$ for $j=1,\ldots,i$, and  set $y^{j}_{s_1} = y^{i}_s$  and $y^{j}_{s_2} = y^{j}_{s} - y^{i}_s$  for $j=i+1,\ldots,N$.
The timesteps $s_1$ and $s_2$ now occur consecutively in the discrete time horizon where $s$ occurred previously; any demand point interval $I(i,t)$ that originally contained $s$ will now contain $s_1$ and $s_2$. Clearly we have maintained feasibility of the LP solution as well as its batch-structure. We can then apply this construction recursively to $s_1$ and $s_2$. Overall, each  distinct value in $\{y^i_s : i \in [N]\}$ corresponds to a timestep generated.
The number of timesteps is now at most $NT$, which is still polynomial in the input. For convenience, we will still let $T$ denote the number of timesteps. 
\item[2.] \textbf{Rounding Candidate Orders via Piping} For each demand point $(i,t) \in D_2$, by the definition of $D_2$, $I(i,t)$ lies in the interior of the interval between 2 consecutive initial general orders in $\mathcal{T}$. For each timestep $s$ with $y_s>0$, we define a candidate (replenishment) order, consisting of all item types with $y^i_s >0$. Once again, if $y^i_s > 0$, then $y^i_s = y_s$. A \textit{non-integral} candidate order is one where $0 < y_s < 1$. Increasing a non-integral order by $\delta$ refers to increasing $y_s$ and each positive $y_s^i$ by $\delta$. An \textit{integral} order is one where $y_s \in \{0,1\}$. We loop over steps $a) \text{ and }b)$ until only two fractional candidate orders remain. Then we round the two corresponding $y_s$ variables up to $1$ and proceed to step $3.$ First, we observe that the following fact is a direct consequence of our definition of candidate orders.

\textbf{Fact:} $\sum_{d \in D} r_d = 1 - \sum_{o} n_o y_o$ where $y_o$ is the value of $y_{s(o)}$ where $s(o)$ is the timestep corresponding to candidate order $o$ and $n_o$ is the number of demand points $(i,t)$ such that item $i$ is included in order $o$ and $I_{(i,t)}$ includes $s(o)$.

\end{itemize}
\begin{itemize}
\item[(a)] We construct a linear system of equations as follows. The components of our solution vector will specify a perturbation each non-integral candidate order so that we can move towards an integral solution without increasing the total cost. 
For each non-integral candidate order $o$, there will be a corresponding variable $\delta_o$ in the linear system. A solution to the system will specify a feasible perturbation to the current solution, and it will be possible to move towards an integral solution without increasing the total cost. The first equation is $\sum_{o} n_o\delta_o = 0$. For every batch $B$ with multiple non-integral candidate orders, we also add the constraints $\sum_{o \in B} \delta_o = 0$. 

\begin{lemma}
If we let $n_o$ be the number of demand points candidate order $o$ serves, the interval over which $H^{it}_s=0$ for each demand point is contained within a single batch and $\sum_{o \in B} y_o \leq 1$, then $\sum_{d \in D} r_d = |D|-\sum_{o} n_o y_o$
\label{pipage_monochromatic_rejection_batch_candidate}
\end{lemma}

\begin{proof}
Until the final round up of $y_s$ variables, for each demand point $d$ with item type $i$, $\sum_{s \in I_d} y^i_s \leq \sum_{s \in I_d} y_s \leq 1$. Therefore, $r_{d} = 1 - \sum_{s \in I_d} y^i_s = 1- \sum_{o \in I_d: i \in o} y_o$ where the last equality follows from $y^i_s \in \{0,y_s\}$ 

$\sum_{d \in D} r_d = \sum_{d \in D} 1 - \sum_{o \in I_d: i \in o} y_o = |D| - \sum_{d \in D} \sum_{o \in I_d: i \in o} y_o = \sum_{o} \sum_{(i,t) \in D : i \in o \in I_{(it)}} y_o = |D| - \sum_{o} n_o y_o$.  
\end{proof}

\item[(b)] We find a non-zero solution to the above system of equations. Let $d_{\ord}$ be the set of candidate orders that serve demand $d$. If, for every candidate order $o$, increasing the corresponding $y_s$ and each $y^i_s$ such that $y^i_s=y_s$ at rate $\delta_o$ does not lead to a net increase in total cost, then do so and simultaneously decrease $r_d$ at a rate $\sum_{o \in d_{\ord}} \delta_o$ until some variable becomes integral. Else, increase the corresponding $y_s$ and all $y^i_s$ such that $y^i_s=y_s$ at rate $-\delta_o$ and simultaneously increase $r_d$ at a rate $\sum_{o \in d_{\ord}} \delta_o$ until some variable becomes integral. The number of variables and constraints in the linear system guarantees the following lemma.
\end{itemize}

    \begin{lemma}
Whenever there are 3 or more non-integral candidate orders, there is a non-zero solution to the linear system.
\label{monopipage-sol-exists}
    \end{lemma}

\begin{proof}
Let $G$ be the set of batches $B$ such that there are multiple non-integral candidate orders in $B$. 
There are exactly $|G|+1$ constraints in the linear system. Therefore, whenever $|G|+2$ non-integral candidate orders (and hence $|G|+2$ variables in the linear system) exist, the system has either no solution, or infinitely many solutions. Since $\delta=0$ is a solution, there must be infinitely many non-zero solutions. Using standard linear algebra techniques, we can find such a solution efficiently.

If $|G| \geq 2$, there must be at least $|G|+2$ non-integral candidate orders (because each element of $G$ has more than one) and therefore, at least $|G|+2$ of them must exist. Therefore, if there is no solution, $|G| \leq 1$. and the number of non-integral candidate orders $|G| + 1 \leq 2$
\end{proof}
\begin{itemize}
    \item[3.] \textbf{Rounding item orders} We now round the $y^i_s$ in a similar manner to the $y_s$ earlier. Loosely speaking, one can think of each (item,batch) pair as corresponding to a batch in step $2$ and whenever there are three or more non-integral $y^i_s$ variables, we can perform steps $(a)$ and $(b)$ below. Once there are two or fewer non-integral $y^i_s$ variables we round them up (and round down the corresponding $r_{it}$ variables) obtaining an integral solution. 
    Note that this point, the choice of item orders can be reduced to a knapsack problem; however, our method scales to multiple colors more directly.
    \item[(a)] For each non-integral $y^i_s$, let $n^i_s$ be the number of demand points in this instance of item $i$ serviceable at time $s$. Once again, we solve a linear system of equations with variables $\delta^i_s$ corresponding to these non-integral $y^i_s$. The first constraint is $\sum_{i,s} n^i_s\delta^i_s = 0$. For every batch $B$ and item type $i$ with multiple non-integral $y^i_s$, we add a constraint $\sum_{s \in B} \delta^i_s = 0$. For similar reasons to before, we get the following two lemmas guaranteeing that we can find such a solution and that the sum of $r_{it}$ variables is $|D|-\sum_{i = 1}^N \sum_{s} n^i_s y^i_s$ :
\end{itemize}    
\begin{lemma}
If we let $n^i_s$ be the number of demand points $(i,t)$  such that $s \in I_{it}$ and assume that for all demand points $d$, $I_d$ is contained within a single batch and $\sum_{o \in B} y^i_s \leq 1$, then the total number of demand points rejected is $|D|-\sum_{i=1}^N \sum_{s} n^i_s y^i_s$
\end{lemma}

\begin{proof}
Notice that this is true at the start of step 3. because until the rounding up of $y_s$ variables, $n_o = \sum_{i=1}^N
n^i_s$. Moreover, the final roundup of $y_s$ variables did not alter any of the $y^i_s$ and $r_{it}$ variables. Therefore, we may assume this is true initially.

For each demand point $d$ with item type $i$, $\sum_{s \in I_d} y^i_s \leq \sum_{s \in I_d} y_s \leq 1$. Therefore, $r_{d} = 1 - \sum_{s \in I_d} y^i_s$ 

$\sum_{d \in D} r_d = \sum_{d \in D} 1 - \sum_{s \in I_d: i \in o} y^i_s = |D| - \sum_{d \in D} \sum_{o \in I_d: i \in o} y_o = \sum_{s} \sum_{(i,t) \in D : i \in s \in I_{(it)}} y^i_s = |D| - \sum_{s \in I_{it}} n^i_s y^i_s$ 
\end{proof}
\begin{lemma}
\label{monopipage-sol-exists2}
Whenever there are 3 or more non-integral $y^i_s$ variables, there is a non-zero solution to the linear system. 
\end{lemma}
\begin{proof}
Similar to before, let $G$ be the set of item, batch pairs $(i,B)$ such that $\sum_{s \in B} y^i_s=1$. There are exactly $|G|+1$ constraints.

Therefore, whenever $|G|+2$ non-integral $y^i_s$ variables (and hence $|G|+2$ variables in the linear system) exist, the system has either no solution, or infinitely many solutions. Since $\delta=0$ is a solution, there must be infinitely many non-zero solutions. Using standard linear algebra techniques, we can find such a solution efficiently.

If $|G| \geq 2$, there must be at least $|G|+2$ non-integral candidate orders (because each element of $G$ has more than one) and therefore, at least $|G|+2$ of them must exist. Therefore, if there is no solution, $|G| \leq 1$. and the number of non-integral $y^i_s$ variables, $|G| + 1 \leq 2$.
\end{proof}
\begin{itemize}
    \item[(b)] Similarly to step $2(b)$, we change $y^i_s$ variables at rate $\delta^i_s$ until some $y^i_s$ becomes integral. We also change $r_{it}$ at rate $-\sum_{s \in I_{it}} \delta^i_s$.
    
\end{itemize}

The cost guarantee below follows from the fact that the only steps that increase the cost of the current solution are the roundups in Steps 2 and 3. All that remain to complete the proof of Theorem \ref{thm: pipage_special}. is to verify feasibility, which is done in Lemma \ref{feasibility_monochromatic}

\begin{lemma}\label{monopipage-last-lemma}
The total extra cost incurred by pipage based rounding is at most $2K_0+2K^{M}_{\max} \leq \frac{\epsilon}{6}\LP^{\sol}_2$.
\end{lemma}

\begin{lemma}
Once a variable becomes integral in the solution maintained by the pipage rounding algorithm, it never changes. Moreover, the solution is always feasible.
\label{feasibility_monochromatic}
\end{lemma}

\begin{proof}
If a $y$ variable becomes integral, clearly it is never changed again because it is never rounded up or down and is never in the linear system.

If an $r_{it}$ variable becomes $1$. Then for all $s$ such that $H^{it}_s = 0$, $y^i_s=0$ and because the corresponding $y^i_s$ variables are not in any subsequent linear system of equations, the $r_{it}$ variable is never changed again by the algorithm.

If an $r_{it}$ variable becomes $0$, then $\sum_{s \in B(it), s \leq t} y^i_s \geq 1$ and the choice of constraints in the linear systems will ensure that this sum never changes. Therefore, $r_{it}$ will remain $0$ until termination.

The above properties imply non-negativity and the construction of the linear system guarantees that the ``serve or reject" constraint is satisfied. Moreover, $y^i_s \leq y_s$ is trivially true.

The only constraint that needs to be checked is the rejection limit. The first constraint in the linear system implies that this must be true as long as $\sum_{s \in B} y^i_s \leq 1$ for all batches $B$. By our choice of constraints in the linear system, we maintain this invariant until the final round up in step 3. Therefore, the rejection limit is satisfied until the final round up in step 3. Since the final round up in step 3 only decreases $r$ variables, the rejection limit is satisfied.
\end{proof}

\subsection{$\rjrpd$: Solving Instance 1 with Iterative Rounding}

Given an LP solution $(\overline{y}^*,\overline{r}^*)$ for Instance 1, we build an iterative rounding algorithm that can find a $(1+\frac{\epsilon}{4})$ approximate integral solution. We will rely on a key structural property of an extreme point optimal solution; at most one item type $i$ has non-integral $\overline{y}^{*i}_s$. Moreover, between the first and last timesteps $s$ such that $0<\overline{y}^{*i}_{s}<1$ (denoted $s_1$ and $s_2$, respectively) there is no $s'$ such that $\overline{y}^{*i}_{s'}=1$. At each iteration, we will round up a single $\overline{y}^{*i}_s$ variable for $s \in [s_1,s_2]$ and before re-solving, update an extra family of constraints that depends on the current solution. This ensures that the sum of non-integral $\overline{y}^{*i}_s$ variables reduces by a constant fraction. Once this sum is sufficiently small, we have a pipage-inspired final step to complete the rounding.  

For the rest of this section, we will simplify notation by referring to $(\overline{y}^*,\overline{r}^*)$ as $(y,r)$. In future sections, we will use $\overline{y}^*,\overline{r}^*$ to distinguish between the optimal LP solution to Instance 1 and $\LP^{\sol}$, the optimal solution to the LP relaxation of the unsplit instance.  

\begin{definition}[Multibatch]
For a given LP solution $(y,r)$ for Instance 1, a multibatch is a pair $(i,I)$ of an item type and interval such that 
(i) for all $s \in I$, $y^i_s < 1$; 
(ii) the first and last timesteps of $I$ have $y^i_s>0$;
(iii) the interval $I$ is maximal in the sense that there is no $I' \neq I$ such that $I \subset I'$ and $(i,I')$ satisfying the first two properties.
\end{definition}

Because of the maximality condition, we know that if there are multiple multibatches of the same item type, the intervals corresponding to them must be disjoint.
\begin{lemma}\label{monoiter-delta}
    For each multibatch $(i,I)$, there exists a vector $\delta$ (with entries corresponding to each variable in the LP) such that (1) $\delta =0$ for all integral variables; (2)  $\delta = 0$ for all $y_s^{i'}$ where either $s \notin I$ or $i \neq i'$; (3) increasing all variables by $\delta$ still satisfies all constraints except the rejection limit constraint; 
    and (4) the above properties hold for $-\delta$ as well.
\end{lemma}

\begin{proof}
    Let $\overline{V}$ be the set of intervals $[s,t]$ entirely within the multibatch such that either $[s,t] \in V^i_{\tight}$ or there exists some demand point $(i,t) \in D_{\tight}$ such that $s(i,t)=s$. Let $\kappa$ be an arbitrarily small positive value. 
    
    Notice that if there are two intervals $I_1=[s_1,t_1],I_2=[s_2,t_2] \in \overline{V}$ and $I_2 \subset I_1$ then $\sum_{s \in [s_1,s_2) \cup (t_2,t_1]} y^i_s=0$ because  $ \sum_{s \in I_1} y^i_s = \sum_{s \in I_2} y^i_s =1$
    Hence, the first and last timesteps in $[s_2,t_2]$ such that $y^i_s >0$ are also the first and last timesteps in $[s_1,t_1]$ such that $y^i_s >0$.

    Define an increasing sequence of timesteps as follows: $t_1$ is the first timestep in the multibatch with positive $y^i_s$. Let $t_2$ be the last timestep $s \in I$ such that $1>y^i_s>0$ and every interval in $\overline{V}$ that contains $t_1$ also contains $s$ (if any such intervals exist; if not the sequence just has one element).

    Given $t_k,t_{k-1}$, we define $t_{k+1}$ to be the last timestep $s \in I$ such that $1>y^i_s>0$ and every interval in $\overline{V}$ that contains $t_k$ but not $t_{k-1}$ also contains $t_{k+1}$ (if no such interval exists, the sequence terminates at $t_k$).

    We will increase $y^i_s$ at all odd members of the sequence and decrease $y^i_s$ at all even members of the sequence (i.e., $\delta=\kappa$ for odd $y^i_s$ and $-\kappa$ for even $y^i_s$). 
    
    Increase (or decrease) the rejection/satisfaction variables for each demand point appropriately; notice that if $r_{it} \in \{0,1\}$, it remains unchanged because it will include either 0 or 2 consecutive members of the sequence.

    Clearly, this satisfies the constraints that $\sum_{s \in I(i,t)} y^{i}_s+r_{it} \geq 1$ while also preserving non-negativity as well as the above properties 1 and 2. Notice that we could make exactly the same arguments for $-\delta$ instead of $\delta$ by decreasing for odd members of the sequence and increasing for even members of the sequence.
\end{proof}

Using this, we can show that at an extreme point optimal solution, there is at most one multibatch:

\begin{lemma}
If $(y,r)$ is an extreme point solution, then there is at most one multibatch $i,I$ in the corresponding solution.
\end{lemma}
\begin{proof}
Suppose for contradiction that there are multiple multibatches. Choose two, $i_1,I_1$ and $i_2,I_2$. By the previous lemma, we can find  $\delta^1,\delta^2$ corresponding to each multibatch. Using properties (3) and (4) from the previous lemma, and the fact that the feasible region of a linear program is convex, we obtain that for any pair of real numbers $\alpha_1,\alpha_2$ such that $|\alpha_1|+|\alpha_2| \leq 1$, the solution $(y',r')=(y,r)+\alpha_1 \delta_1 + \alpha_2 \delta_2$ satisfies all constraints except the rejection limit and that the same is true for $(y',r')=(y,r) -
\alpha_1 \delta_1 -
\alpha_2 \delta_2$.

However, it is clear that we can choose $\alpha_1,\alpha_2$ such that $\alpha_1 \sum_{d \in D} \delta^1_d + \alpha_2 \sum_{d \in D} \delta^2_d = 0 = -\alpha_1 \sum_{d \in D} \delta^1_d - \alpha_2 \sum_{d \in D} \delta^2_d$ and $|\alpha_1|,|\alpha_2| \leq 1$.

For these values of $\alpha_1,\alpha_2$ both $(y',r')=(y,r)+\alpha_1 \delta_1 + \alpha_2 \delta_2$ and $(y',r')=(y,r) - \alpha_1 \delta_1 -\alpha_2 \delta_2$ both satisfy the rejection limit because $\sum_{d} r'_d = \sum_{d} r_d$ in both cases.

Therefore, they are both feasible
\end{proof}
If there are multiple multibatches, the $\delta$ vectors corresponding to the multibatches would be linearly independent and there exists a non-zero linear combination of them that keeps the total of the rejection variables unchanged. This would imply that the current solution is not an extreme point. Therefore, an extreme point solution can have at most one multibatch.

\smallskip
\noindent
{\bf The Iterative Rounding Algorithm}
First, recall that we construct $\overline{y}$ in the same manner as the $2-$approx-imation for $\jrpd$. Then we choose $r_{it} \geq 0$ to be as small as possible while also satisfying the ``every demand must be served or rejected'' constraint. Each rejection variable that takes value $0$ (or $1$) is constrained to remain $0$ (or $1$ respectively) for the rest of the algorithm. Consequently, for any demand with $r_d=0$ (or $r_d=1$ respectively), it follows that $\overline{r}_d=0$ (or $\overline{r}_d = 1$ respectively) even for the rounded solution $(\overline{y},\overline{r})$.
\begin{itemize}
    \item[1.] \textbf{Finding the multibatch} We first find an extreme point optimal solution $(y,r)$ to the LP for Instance 1. Let $i$ be an item type with some $s$ such that $y^i_s$ is non-integral (if none exists, then our solution is integral and no rounding is required). Let $\LP_{\itm}$ be the item ordering cost of this solution. The figure below shows an example of what a multibatch might look like.
  
\end{itemize}

\begin{figure}[h]
\includegraphics[width=12cm]
{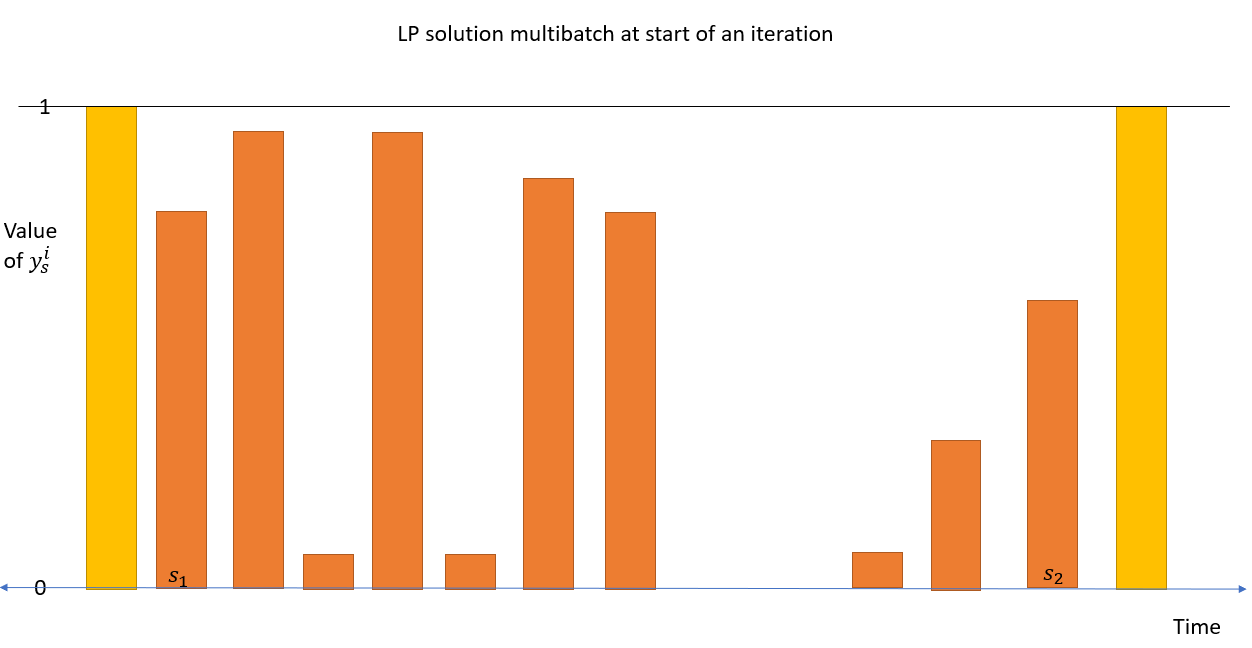}
\centering
\end{figure}
\begin{itemize}
    \item[2.] \textbf{Reducing Fractional Weight} Let $s_1$ and $s_2$ be the first and last timesteps with non-integral $y^i_s$. Let $Z^i_t=\sum_{s=s_1}^t y^i_s$.  If we already have $Z^i_{s_2} \leq 4$, proceed to Step 4. Else, for every sub-interval $I$ of $[s_1,s_2]$ such that $\sum_{s \in I} y^i_s \geq 1$, add the constraint $\sum_{s \in I} y^i_s \geq 1$ to the LP. Add constraints fixing the value of all integral variables. The figure below shows some interval constraints that will be added.
\end{itemize}

\begin{figure}[h]
\includegraphics[width=12cm]
{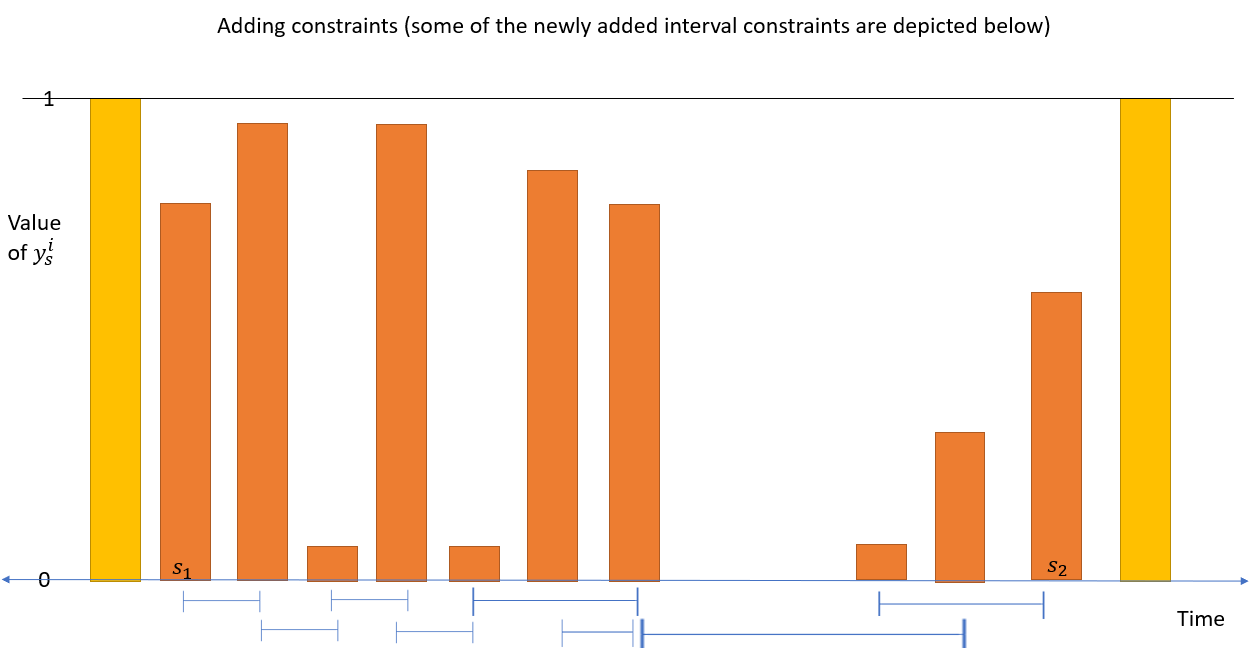}
\centering
\end{figure}

\begin{itemize}
    \item[3.] \textbf{Rounding up some variable(s)} Let $t'$ be such that $Z^i_{t'-1}<\frac{Z^i_{s_2}}{2}\leq Z^i_{t'}$. Round $y^i_{t'}$ up to 1 and constrain $y^i_{t'} =1$. Resolve the LP and go back to Step 2. Due to the extra constraints we added, a significant portion of the ordering cost remains on each side (before and after) of the rounded up variable, but the next multibatch will be either entirely before or after the rounded up timestep. Consequently, the fractional weight will reduce. This is formalized in lemma \ref{iterative_special_logarithmic_iterations} below. The figures preceding the lemma show what the solution looks like immediately before and immediately after the re-solving of the LP in this step.
\end{itemize}
\textbf{Note:} Yes, we are (almost immediately) discarding the LP solution given as input and instead finding an optimal extreme point LP solution. It is only almost immediate because we added constraints ensuring that the rejection variable corresponding to a demand point remains unchanged if it was $0$ (or $1$) in the integral solution. Notice that forcing $r_d=1$ is essentially the same as deleting a demand point and reducing the rejection limit by $1$. Constraining $r_d=0$ for a demand point is the same as adding an interval constraint (with interval $I_d$) that must be satisfied. 

This intuition of the interval constraints as being artificial demand points that cost nothing extra to service may make the algorithm and its analysis more digestible.

\begin{figure}[h]
\includegraphics[width=12cm]
{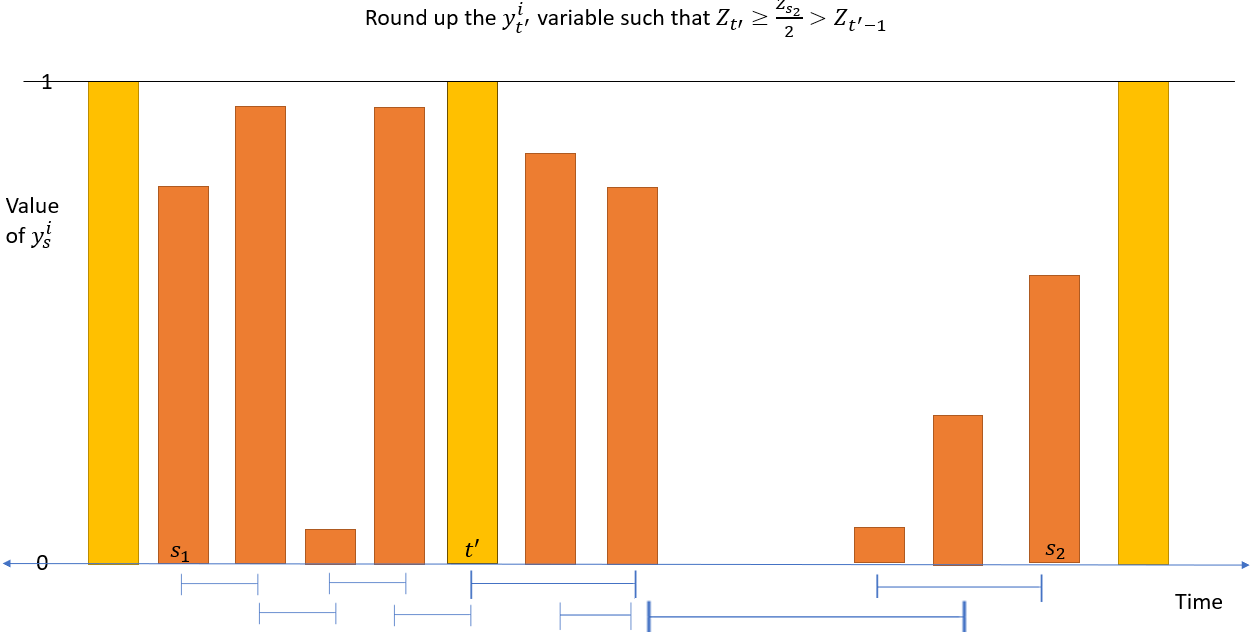}
\centering
\end{figure}

\begin{figure}[h]
\includegraphics[width=12cm]
{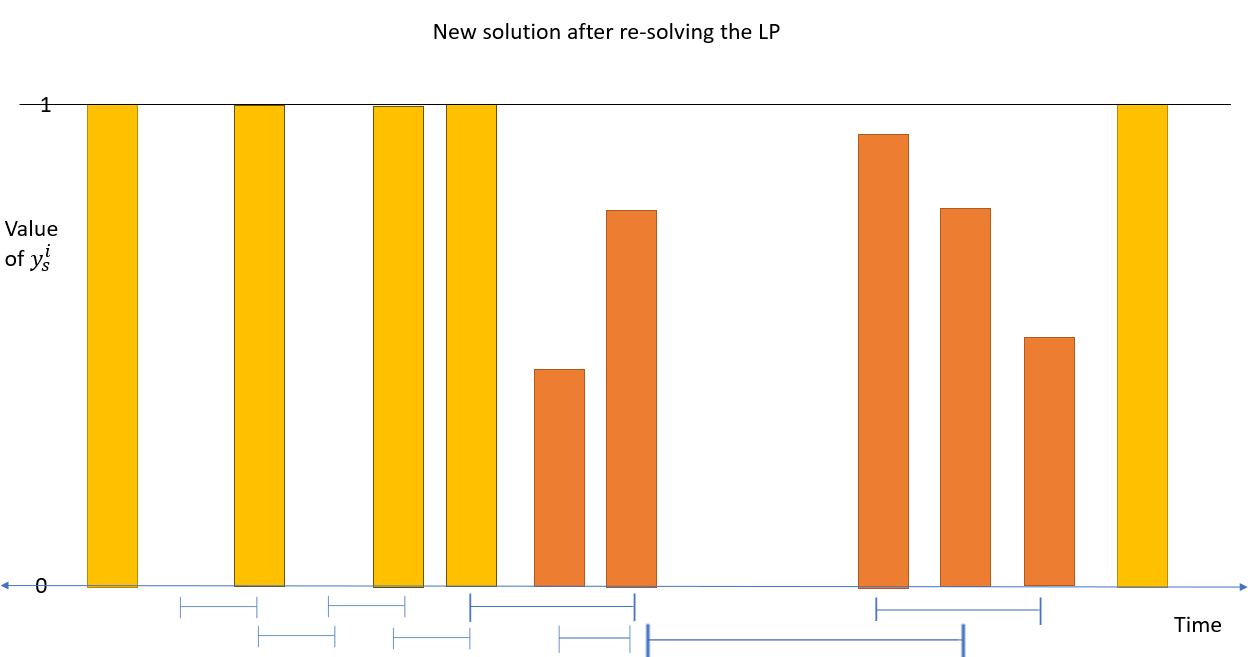}
\centering
\end{figure}
    \begin{lemma}
    In each iteration of Steps 2 \& 3 of the iterative rounding procedure, the total fractional cost $\sum_{s:y^i_s < 1} y^i_s$ decreases by a factor of at 
 least 8/7. Consequently, if $Q_{\init} > 4$ is the initial fractional weight $Z^i_{s_2}$, then $\lceil\log_{\frac{8}{7}}(\frac{Q_{\init}}{4})\rceil \leq 10\ln(Q_{\init}) - 9$ many iterations suffice.
    \label{iterative_special_logarithmic_iterations}
    \end{lemma}

\begin{proof}
    Let $(i,I)$ be the multibatch with $I = [s_1,s_2]$. We round $y^i_{t'}$ up to 1 where $t'$ satisfies $Z^i_{t'-1}<\frac{Z^i_{s_2}}{2}\leq Z^i_{t'}$. Let $Q = Z^i_{s_2}$.


Since $y^i_{s} \leq 1$ for all $s$, $\sum_{s=s_1}^{t'-1} y^i_s$ and $ \sum_{s=t'+1}^{s_2} y^i_s$ are both at least $\frac{Q}{2}-1$. Now consider the family of intervals $\mathcal{J}$ as follows.


Start with the interval $J= [ s_1 , s_1 + 1]$ . If $\sum_{s \in J} y^i_{s} <1$ then increase the right endpoint of $J$ by 1. Repeat until $\sum_{s \in J} y^i_{s} \geq 1$ and then add $J$ to $\mathcal{I}$. To find the next interval $J'$, start with the first timestep $t$ such that $[t,t+1] \subset [s_1,t'-1]$, $y^i_t > 0$ and $t$ is not a part of an interval in $\mathcal{J}$ (if such a $t$ exists). If we can extend this interval $J'$ in the same manner as before while ensuring that $\sum_{s \in J'} y^i_{s} \geq 1$, do so and add $J'$ to $\mathcal{J}$. Otherwise terminate the procedure.


Notice that every interval $J \in \mathcal{J}$ satisfies $1 \leq \sum_{s \in J} y^i_{s} < 2$. Therefore, the LP solution in the next iteration will ensure that $\sum_{s \in J} y^i_{s} \geq 1$ for all $J \in \mathcal{J}$. Moreover, the intervals in $\mathcal{J}$ are pairwise disjoint and $\sum_{J \in \mathcal{J}} \sum_{s \in J} y^i_{s} > \frac{Q}{2}-2$.

Therefore, there are at least $\lceil \frac{Q}{4} - 1 \rceil$ many such intervals in $[s_1,t'-1]$. Note that $\lceil \frac{Q}{4} - 1 \rceil \geq 1$ for $Q > 4$. Similarly, there are at least $\lceil \frac{Q}{4}  -1 \rceil$ such intervals in $[t'+1,s_2]$. When we re-solve the LP with $y_{t'}^i = 1$, there can once again be at most one multibatch $(i,I')$. Further, $I'$ must be a sub-interval of either $[s_1,t'-1]$ or $[t'+1,s_2]$, so either $[s_1,t'-1]$ has no non-integral $y^i_{s}$ variables or $[t'+1,s_2]$ has no non-integral $y^i_{s}$. However, we argued that both $[s_1,t'-1]$ and $[t'+1,s_2]$ have $\lceil \frac{Q}{4}-1 \rceil$ disjoint intervals $J$ constrained to have $\sum_{s \in J} y^i_{s} \geq 1$. Therefore, the sum of fractional $y^i_{s}$ variables must have decreased by at least $\lceil \frac{Q}{4}-1 \rceil$. Therefore, the new fractional weight is at most $Q + 1 - \lceil \frac{Q}{4} \rceil \leq \frac{7Q}{8}$ for $Q > 4$. 

The claim about the number of iterations directly follows. 
\end{proof}
\begin{itemize}     \item[4.] \textbf{Finishing up the last non-integer solution} For all $k \in \{0,\ldots,4\}$, let $t'_k$ be the first timestep in which $Z^i_{t'_k}$ exceeds $k$. Define $t'_{4}$ to be $s_2$ (note that $t_0$ is $s_1$). For $t'_0,\ldots,t_{4}'$, round $y^i_{t'}$ to 1.
    \item[5.] Apply pipage rounding Step 3 with the remaining non-integral $y^i_s$ variables.
\end{itemize}

By Lemma \ref{iterative_special_logarithmic_iterations}, the number of iterations is logarithmic in the total fractional orders placed. However, our choice of $M$ ensures that when the number of fractional orders exceeds $M$, the ratio $\frac{\log Q_{\init}}{Q_{\init}}$ is small. If the number of fractional orders is below $M$, we can bound the cost of rounding by using the guessed components of our solution. This gives us the following lemma, which completes the proof of Theorem \ref{thm:iterative_special}:

\begin{lemma}\label{monoiter-final}
The total extra cost incurred by iterative rounding is at most $\frac{\epsilon}{4} \LP_{\itm}$.
\end{lemma}

\begin{proof}
Suppose $Q_{\init} > 4$. By iterative rounding, we incur an additional cost of $(10\ln(Q_{\init}) - 9)K_i$. We then open $5$ orders of cost $5K_i$ before the pipage step, which incurs an additional cost of $2K_i$. So we incur an additional cost of $7K_i$. Therefore, the total cost is at most $10\ln(Q_{\init}) K_i$. 
Since we guessed the $M$ most expensive item orders, we have the following lower bounds on $\LP_{\itm}$:
\[\LP_{\itm} \geq M K_i, Q_{\init} K_i\]
\begin{itemize}
\item If $Q_{\init} \leq M$, $10\ln(Q_{\init}) K_i \leq 10\ln(M) K_i \leq \frac{10\ln(M)}{M} \LP_{\itm}$.
\item If $Q_{\init} > M$, $10\ln(Q_{\init}) K_i \leq \frac{10\ln(Q_{\init})}{Q_{\init}} \LP_{\itm} \leq \frac{10\ln(M)}{M} \LP_{\itm}$ as $\frac{\ln x}{x}$ is decreasing for $x \geq 3$ and $M \geq 3$.
\end{itemize}
Suppose $Q_{\init} \leq 4$. Then the cost is at most $7K_i \leq 10 K_i \leq \frac{10}{M} \LP_{\itm} \leq \frac{10\ln M}{M} \LP_{\itm}$. We chose $M = 1000 + \frac{1000}{\epsilon}\ln(1000/\epsilon)$, so $\frac{10\ln M}{M} \leq \frac{\eps}{6}$ as desired. 
\end{proof}

\section{General Holding Costs and Fairness Constraints}
Our techniques are flexible enough to accommodate general holding costs as well, and almost directly yield a $(4+\epsilon)$ approximation for $\rjrp$. In this section, we first sketch this result, and then provide the main ideas for improving upon the constants, as well as extending the results to our fairness settings (though, as indicated initially, we will restrict the presentation here to the setting in which the demands are partitioned into color classes.)

\textbf{General Holding costs,} $\mathbf{4+\epsilon:}$ We reduce $\rjrp$ to $\rjrpd$ by solving the LP relaxation for $\rjrp$ and letting, for each demand $d=(i,t)$, $I(d)=[s(d),t]$ where $s(d)$ is the earliest timestep with $x^d_s>0$; that is, $H^{d}_s := 0$ if $s \in I(d)$ and $H^{d}_s := \infty$ if $s < s(d)$. Then we apply the $(3+\epsilon)$-approximation algorithm. If the resulting solution rejects fewer than $R$ demand points, we reject additional demand points with $r_d>0$ in $\LP^{\sol}$ so that there are exactly $R$ rejections. We will use the fact that if $r_d=0$ for any demand point, then the demand point must be serviced (it must be true that $\sum_{s \in I(d)} y_s \geq 1$ and therefore, there must be an IGO placed in the interval $I(d)$). The dual LP is shown below:

\begin{empheq}{alignat=3}
\text{minimize}\quad &  \sum_{(i,t) \in D} b_{it} - R \lambda \\
\text{subject to } & b_{it} - l^{it}_s \leq H^{it}_s \quad \text{for each } (i,t) \in D; \\
& \sum_{(i,t) \in D : t \geq s} l^{it}_s - z^i_s \leq K_i \quad \text{for each } i \in [N], s \in [T]; \\
& \sum_{i=1}^N z^i_s \leq K_0 \quad \text{for each } s \in [T]; \\
& b_{it} - \lambda \leq 0 \quad \text{for each } (i,t) \in D; \\
& b_{it},l^{it}_s,z^i_s,\lambda \geq 0 \quad \text{for each } s \in [T], i \in [N], (i,t) \in D 
\end{empheq}

$(2)$ corresponds to the $x^{it}_s$ variables, $(3)$ to the general ordering variables, $(4)$ to the item ordering variables, $(5)$ to the rejection variables. The $b_{it}$ variables correspond to the "serve or reject" constraint, $\lambda$ to the rejection limit constraint, $l^{it}_s$ to the "servicing requires an item order" and the $z^i_s$ variables to the "item order requires a general order" 

We use the optimal dual solution to show that the holding costs are at most $\opt$. The objective value of the optimal dual solution is $\sum_{d \in D} b_d^* - R\lambda^* , $ where $b_d^*$ is the value of the dual variable corresponding to the ``serve or reject'' constraint for demand $d$ and $\lambda^*$ is the value of the variable corresponding to the rejection limit. By complementary slackness, $x^d_s>0 \implies b_d \geq H^d_s$ and $r_d>0 \implies b_d = \lambda^*$. However, the only demands rejected have $r_d>0$, which implies that if $A$ is the set of serviced demand points, the dual objective value is $\sum_{d \in A} b_d^*$. This implies, by strong duality, that the total holding cost incurred by serviced demand points is at most $\opt$. Note that even this bound improves on previously results for $\rjrp$.

We outline how to refine and extend this approach to achieve the approximation factors claimed and in the generality described in the introduction. In this section we hope to illustrate the main concepts without miring the reader in technical details, which are then given in subsequent sections.
\vspace{-.4in}
\paragraph{Weights} If each demand point $d$ has a weight $w_d$, we can simply solve the LP replacing the $\sum_{d \in D} r_d \leq R$ constraint with $\sum_{r \in D} w_dr_d \leq R$. We split into 2 instances and generate a fractional solution in exactly the same way. During pipage rounding, we will now simply ensure the total weighted fractional sum of rejections does not decrease; during iterative rounding, having weights does not change our procedure of rounding up service variables and resolving.
\vspace{-.2in}
\paragraph{Improving the constant}
In the discussion that follows, we will use $(x^*,y^*,r^*)$ to denote the optimal LP solution prior to splitting (and the initial solution to Instance 2) and $(\overline{x},\overline{y},\overline{r})$ to denote the initial solution to Instance 1.

\begin{itemize}
    \item We will show in the following sections that modified versions of our rounding algorithms yield $(1+\epsilon)$-approximation factors to Instances 1 and 2 even with general holding cost functions.
    \item  Notice that we can take all rejected demands from Instance 1 and add them to Instance 2 (but deleting all $x^{it}_s$ not in the batch containing their deadline) and the pipage rounding algorithm will still yield a solution with cost at most $(1+\epsilon) \opt$.
    \item We can think of the rejection limit $\sum_{d \in D} w_dr^*_d \leq R_{\weight}$ as being equivalent to a satisfaction requirement $\sum_{d\in D} w_d(1-r^*_d) \geq W-R_{\weight}$ where $W=\sum_{d \in D} w_d$.
    
    \item Suppose we found a set of initial general orders $\mathcal{T}$ in the same way as in Section 3. For a demand point $d = (i,t)$, let $o(d)$ be the last timestep $s \in \mathcal{T}$ such that $s \leq t$. We introduce a nonlinear integer program below, where the variables $x^{it}_{\Left}$ and $x^{it}_{\Right}$ represent the extent to which $(i,t)$ is served in Instance 1 and 2, respectively. Notice that any integral solution to this nonlinear integer program is a feasible solution for the original problem and has the same cost.    
    \begin{empheq}{align*}
\text{minimize} \quad & \sum_{s=1}^T K_0 y_s + \sum_{i=1}^N \sum_{s = 1}^T K_i y^i_s + \sum_{(i,t) \in D} \sum_{s=1}^T H_s^{it} x_s^{it} + \sum_{(i,t) \in D} p_{it}r_{it}
\end{empheq}
\begin{empheq}{align*}
\text{subject to }
\quad & x^{it}_{\Left}=\sum_{s \in \mathcal{T}: s \leq o(d)} x^{it}_s , \quad \text{for each } (i,t) \in D;\\
\quad & x^{it}_{\Right}=\sum_{s: o(d) < s  \leq t} x^{it}_s , \quad \text{for each } (i,t) \in D;\\
\quad & y^i_s \leq y_s, \quad \text{for each } i \in [N], s \in [T]; \\
\quad & x_s^{it} \leq y_s^i , \quad \text{for each } (i,t) \in D, s \in [T]; \\
\quad & \sum_{s \in [s_1+1,s_2-1]}y_s \leq 1, \quad \text{for each pair of consecutive IGOs } s_1,s_2 \in \mathcal{T};\\
\quad & \sum_{(i,t) \in D} (x^{it}_{\Right}+(1-x^{it}_{\Right})x^{it}_{\Left})w_d \geq W-R; \\
\quad & y_s = 1, \quad \text{for each } s \in \mathcal{T};\\
\quad & y_s^i,y_s,x^{it}_s, x^{it}_{\Left}, x^{it}_{\Right} \in \{0,1\}, \quad \text{for each } (i,t) \in D, s \leq t,s \in [T].
\end{empheq}
    
    \item Notice that given a non-integral solution to the above nonlinear program, rounding to an integral solution is relatively simple. Let $W=\sum_{d \in D} w_d$ and $ W'=\sum_{d \in D} (1-x^d_{\Right})w_d $. Fix $x^d_{\Right}$ and all corresponding variables;$x^d_{\Left}$ is now a solution to Instance 1 (with $r^d_{\Left}=1-x^d_{\Left}$ and the caveat that for each demand $d$, $w_d$ is changed to $(1-x^d_{\Right})w_d$ and the rejection limit is now $W'-(W-R-\sum_{d \in D} w_d x^d_{\Right})$). We round $x^d_{\Left}$ by iterative rounding. Then, fix the integral $x^d_{\Left}$ and by treating $x^d_{\Right}$ as a solution to Instance 2, round $x^d_{\Right}$. Both these rounding procedures cost a factor of $(1+\frac{\epsilon}{4})$ and yield an integral solution.

    \item  All that remains is to find a good non-integral solution. For $\rjrpd$, instead of placing general orders when $Z_t$ reaches $1,2...$ we can place general orders when $Z_t$ exceeds $\frac{1}{3},\frac{2}{3}...$. The advantage of doing this is that when constructing a feasible LP solution for Instance 1, instead of choosing $y^i_s$ at a timestep $s$ with an initial general order to be the sum of $y^{i*}_s$ in the batches immediately preceding and immediately succeeding the initial general order, we can simply choose it to be the sum of $1.5y^{i*}_s$ in the batch immediately preceding the general order. This will give us a feasible solution to the nonlinear problem above which we know we can round. This solution has cost $(4+\epsilon)\LP^{\sol}_{\gen}+(2.5+\epsilon)\LP^{\sol}_{\itm}$. The best of this and the algorithm as described earlier is a 
    $(2.8+\epsilon)$-approximation.
    \item 
    \textbf{General holding costs improvement} We now describe how to obtain an integral solution of cost $3\LP^{\sol}_{\gen}+3\LP^{\sol}_{\itm}+2\LP^{\sol}_{\hold}$. A more careful choice of parameters (the correct choice will depend on the relative ratios of $\LP^{\sol}_{\gen},\LP^{\sol}_{\itm},\LP^{\sol}_{\hold} $) gives us a $\frac{3\sqrt{5}-1}{2}+\epsilon = (2.86+\epsilon)$-approximation. 
    
    We place orders when $Z_t$ exceeds $0.5,1,1.5,2...$. Item ordering variables' values are chosen in the same way. However, when building $x^{it}_s$ values for Instance 1, we will need to be more cautious. For each initial general order timestep $o$, let $o^+$ be the set of timesteps in the batch after $o$ and $o^-$ be the set of timesteps in the batch before $o$. $x^{d}_o=\sum_{s \in o^-} x^d_s$. Then we let $o_d$ be the last initial general order before the deadline of $d$. We increase $x^d_{o_d}$ by $\min(d_{\Left},d_{\Right})$. This results in a feasible solution to the above non-linear program with cost $(3+\epsilon)\LP^{\sol}_{\gen}+(3+\epsilon)\LP^{\sol}_{\itm}+(2+\epsilon)LP^{\sol}_{\hold}$.
\end{itemize}
\noindent
\textbf{Multiple colors}
This step shows the power of the general framework in which we rely on pipage rounding: we can solve a linear system of equations to find a direction that does not increase cost and ensures all color rejection constraints are simultaneously satisfied with every batch having $\sum_{s} y_s \leq 1$. The system has a column for each candidate order and a row for each color and a row for each batch with multiple candidate orders. We show that if there are more than $2C$ fractional candidate orders, we can find a such a direction. Once there are at most $2C$ non-integral candidate orders, we round the $y_s$ variables and re-apply pipage rounding with each item order now being a candidate order and each item-batch pair as a separate batch. This gives us at most $2C$ non-integral item orders, all of which we round up. 

During iterative rounding, there are now $O(C)$ multibatches -- we split the largest multibatch by placing $O(C)$ integral orders. Moreover, the fraction by which the total non-integral $y^i_s$ variables decreases is $1-1/{O(C)}$ and these complexities require us to round up $O(C^3 \log C)$ many item orders. To bound the cost of this, we choose $M=\theta({C^4}/{\epsilon})$.

\noindent
\textbf{Rejection penalties} Remarkably, our techniques yield exactly the same approximation factors when every demand point also has a (non-negative) rejection penalty. First, we note that nowhere in our analysis or algorithms do we exploit the fact that every demand point has only one color. To achieve the desired approximation, we first solve the LP with the correct guesses. Then we add an artificial constraint that the total penalty due to rejections does not exceed the penalty cost incurred by the initial LP solution. When running the pipage rounding and the iterative rounding algorithms, we treat this artificial constraint as an extra color constraint. Our returned solutions will have rejection penalty cost no greater than the initial LP solution. The remaining components of the cost incurred can be bounded the same way as before. Note that line of argument allows for the immediate extension to the case where the weights are $C$-dimensional feature vectors.

\textbf{Time-dependent general ordering costs:} If the fixed cost of a replenishment order were not fixed, but instead varied arbitrarily with time and for timestep $s$ was $K_s$ instead of $K_0$, then we would need to guess the $M$ timesteps with the most expensive general orders placed in the optimal solution. Fix the corresponding $y_s$ variables to $1$ and the algorithm would be almost exactly the same. 

Instead of placing IGOs when $Z_t$ exceeds $0,1,2...$, we would, uniformly at random choose $\delta \in [0,1]$ and place general orders when $Z_t$ exceeds $0, \delta, 1+\delta, 2+\delta ...$ and at the smallest $t$ such that $Z_t=Z_T$. The expected general ordering cost of the IGOs is at most $(1+\frac{\epsilon}{8})LP^{\sol}_{\gen}$. The additional cost in the pipage rounding algorithm due to rounding up general orders would be at most $2C G_{max}$ and therefore, we would obtain the same constant approximation. 
\section{Algorithm for General Holding Costs}
In this section, we will discuss how our algorithm works with general holding costs and fairness constraints: in this setting, we have $C$ colors and every demand point $d$ has weight $w^c_d$. We are given a rejection weight limit for each color. This section and the ones that follow it will be much more focused on proving that we achieve the desired results and much less focused on providing intuition.

We will assume the rejection penalties $p_d$ are all $0$ and explain how to generalize to arbitrary penalties at the end. This algorithm will be a $3$-approximation for $0-\infty$ holding costs $(\cjrpd)$ and a $4$-approximation for general holding costs $(\cjrp)$. We will explain how to improve these factors to $2.8+\epsilon$ and $2.86+\epsilon$ in a later section.

The algorithm for the general case proceeds as described above. We require a fixed amount of side information which can be found by enumerating over all possibilities. The amount of information required depends on $C$ but crucially, neither on $|D|$ nor on $T$. We require $M$ to be sufficiently large that $g(m) = (C+1)\frac{40C^2+90C^2 \log m}{m} \leq \frac{\epsilon}{8}$ and $g'(m)<0$ for all $m \geq M>\frac{50C}{\epsilon}$. There exists a sufficiently large constant $W$ which guarantees that this occurs when $M=W\frac{C^4}{\epsilon}$. We choose $M$ appropriately and proceed as follows:

\paragraph{Small Solutions} We can solve instances with an optimal solution that has fewer than $M$ replenishment orders by enumeration.

\textbf{Case 1:} At most $M$ demand points are serviced in an optimal solution. If an optimal solution serves at most $M$ demand points, we can, by enumeration, guess both the demand points and the time at which each of them is served in the optimal solution. This means we know the exact set of timesteps at which replenishment orders exist and which item types are in each such order; consequently, we have an optimal solution.

\textbf{Case 2:} More then $M$ demand points are serviced, but at most $M$ item orders are used in an optimal solution. If an optimal solution has at most $M$ item orders, we guess these by enumeration. This means we know the entire replenishment order schedule and the only remaining task is to identify the set of demand points serviced. In this case we guess the $M$ demand points that the optimal solution serviced with the most expensive holding cost and the timestep at which each of these demands was served. Let the smallest of these $M$ guessed holding costs be $H^{M}_{max}$. For all demand points $(i,t)$ other than , If $H^{it}_s > H^{M}_{\max}$, $H^{it}_s := \infty$.
\begin{lemma}
\label{fixed_replenishment_schedule_lemma}
Given the complete replenishment order schedule (i.e., all item ordering variables and general ordering variables are integral) but a (possibly non-integral) $x^d_s,r_d$ such that the corresponding solution is feasible and has cost $Q$, we can find an integral solution of cost at most $Q+C H^M_{\max}$. 
\end{lemma}
\begin{proof}

Each demand point $(i,t)$ will either be served at the last order of item $i$ before $t$ (denoted $q(i,t)$), or be rejected. This means that we essentially need to optimize the following IP (which only has $C$ constraints other than $r_{it} \in \{0,1\}$ ):

\begin{empheq}{align*}
\text{minimize} \quad &  \sum_{(i,t) \in D} H_{q(i,t)}^{it} (1-r_{it})\\
\text{subject to }
\quad & \sum_{(i,t) \in D} w^c_{it}r_{it} \leq R_c \quad \text{for all } c \in [C]\\
\quad & r_{it} \in \{0,1\} \quad \text{for all } (i,t) \in D \\
\end{empheq}

We can simply solve the natural LP relaxation; since at most $C$ variables will be non-integral, we can simply round them down to $0$ (i.e., service the corresponding demand points). This increases cost by at most $CH^M_{\max}$
\end{proof}
\textbf{Case 3:} More than $M$ item orders and demand points serviced but at most $M$ replenishment orders placed in an optimal solution. If an optimal solution has at most $M$ replenishment orders, guess them by enumeration. Then guess the $M$ most expensive item orders and holding costs in the optimal solution. Let $K^M_{\max}$ be the item ordering cost of the $M$th most expensive item order as per our guesses and like before, let $H^M_{\max}$ be the $M$th most expensive holding cost guessed. We add constraints forcing $y_s=1$ if $s$ was a guessed replenishment order timestep and $y_s=0$ otherwise. For any item order $i,s$ not guessed, if $K_i>K^M_{\max},$ then add the constraint $y^i_s=0$. For all the guessed item orders $i,s$ add the constraint $y^i_s=1$. If $H^{it}_s > H^M_{\max}$ and demand point $(i,t)$ apply the iterative rounding procedure described later in this section with initial general orders at the guessed timesteps, which yields a $1+\epsilon$ approximation.

\paragraph{`Big' Solutions:} From now on, we will assume that any optimal solution has at least $M$ replenishment orders and that our side information is correct and proceed very similarly to before:

\begin{itemize}
    \item[(a)] \textbf{Augmenting and solving the LP:} Add constraints setting the general and $M$ most expensive item order variables corresponding to our guesses to $1$. We add constraints forcing the $x^{it}_s$ variables corresponding to the guessed demand points to 1. For all $(i,s)$ not in the guessed set such that $K_i > K^{M}_{\max}$, add the constraint $y^i_s=0$. For all $(d,s)$ such that $d$ was not one of the guessed demand points and $H^d_s>H^{M}_{\max}$, we change $H^d_s := \infty$. We also let $\kappa$ be the smallest positive value of any $H^{it}_s$. We increase each $H^{it}_s$ by $\frac{\epsilon \kappa(t-s)}{|D|T}$. This increment is not necessary, but greatly simplifies the analysis of the iterative rounding algorithm. Solve the LP (it must be feasible if our side information is correct).
    \item[(b)] \textbf{Initial General Orders (IGO):} Let $Z_t=\sum_{s \leq t} y_s$. Let $\mathcal{T}$ be the set of timesteps $t'$ in which the running sum $Z_{t'}$ first reaches $0,1,2,,\ldots, \lfloor Z_T \rfloor, Z_T$. Place a general order at each time $t \in \mathcal{T}$. Using these timesteps, we split the instance into two and give a $(1+\epsilon)$-approximation for each.
    \item[(c)] \textbf{Instance 1:} This consists of timesteps $\mathcal{T}$ and the set of demand points $(i,t)$ such that there exists $s \in \mathcal{T}$ for which $H_s^{it} = 0$; let $D_1$ denote this set. The rejection limit for each color $c$ in this instance is at most the extent to which the LP solution rejected demands in $D_1$: $\sum_{(i,t) \in D_1} w^c_{it}r_{it}$. We solve this instance via iterative rounding by treating $K_0=0$ for each $t \in \mathcal{T}$, since the general orders are already opened. As before, the key structural feature of this instance is that because there is a general order at each timestep $t \in \mathcal{T}$, the item types are (somewhat) decoupled. The only thing connecting them is the shared colorwise rejection limits. This structure enables us to obtain and utilize the following guarantee:
\end{itemize}
    \begin{theorem}
    \label{thm:iterative_general}
    Given an instance of $\cjrp$ with $p_d=0$ for all $d$ and the constraint that replenishment orders are placed at every timestep and an LP solution with cost $\LP^{1}$ we can obtain an integral solution of cost $(1+\frac{\epsilon}{8})\LP^{1}$ provided:
    \begin{itemize}
    \item $M$ is sufficiently large that $g(m) = (C+1)\frac{40C^2+90C^2 \log m}{m} \leq \frac{\epsilon}{8}$ and $g'(m)<0$ for all $m \geq M>\frac{50C}{\epsilon}$.
    \item Every optimal solution has at least $M$ replenishment orders
    \item The LP to which $LP^1$ is a solution is augmented with the necessary side information corresponding to $M$.
    \end{itemize}
    \end{theorem}
\begin{itemize}
\item[(d)] \textbf{Instance 2:} This consists of timesteps $[T]-\mathcal{T}$ and all demand points $D_2$ that are not in $D_1$. The rejection limit for every color $c$ this instance is $\sum_{(i,t) \in D_2} w^c_dr_{it}$. We will add the constraint that between any two consecutive initial general orders $s_1$ and $s_2$, $\sum_{s=s_1+1}^{s_2-1} y_s \leq 1$.  The pipage rounding framework relies on being able to partition the time horizon, non-integral replenishment orders, and, critically, demand points into ``batches'', each of which has $\sum_{s} y_s \leq 1$. Maintaining this invariant ensures that the rejection variables are a linear function (instead of piecewise linear) of the $y$ variables which is essential in obtaining and utilizing the following guarantee:
\end{itemize}
    \begin{theorem}
    \label{thm: pipage_general}
    Given an instance of $\cjrp$ with $p_d=0$ for all $d$, an LP solution $(x,y,r)$ of cost $\LP^{2}$ to the instance and a partition of the time horizon into intervals called ``batches" such that for each batch $B$ $\sum_{s \in B} y_s \leq 1$ and for each demand point $(i,t)$, $x^{it}_s = 0$ for all $s$ not in the same batch as $t$, we can obtain an integral solution of cost $(1+\frac{\epsilon}{8}) \LP^2$ provided:
    \begin{itemize}
    \item $M$ is sufficiently large that $g(m) = (C+1)\frac{40C^2+90C^2 \log m}{m} \leq \frac{\epsilon}{8}$ and $g'(m)<0$ for all $m \geq M>\frac{50C}{\epsilon}$.
    \item Every optimal solution has at least $M$ replenishment orders
    \item The LP to which $LP^1$ is a solution is augmented with the necessary side information corresponding to $M$.
    \end{itemize}
    \end{theorem}
    \textbf{Remark:} One can adjust the algorithms used to prove these theorems to ensure that the general ordering cost, item ordering cost and holding cost are larger than their $\LP^1$ counterparts by at most $(1+\frac{\epsilon}{8})$ with a little more side information. We don't do this because it merely complicates the proof without any additional value in the bound.
\begin{itemize}
    \item[(e)] \textbf{Forming the solution:} We merge the solutions from the two instances; if a replenishment order is in either solution, it is in the merged solution. If a demand point is served at some timestep in its' corresponding instance, it's served in the merged solution.
\end{itemize}

\paragraph{Notation:} For the rest of the paper we will use $s(i,t)$ to denote the earliest timestep that $x^{it}_s>0$ and $o(it)$ to denote the last IGO before $t$. Also, we will often use the word ``shift" to describe how we're constructing the item order and demand point service variables at IGO timesteps. Suppose we have two consecutive IGOs $s_k,s_{k+1}$ and several timesteps between them. What we mean by, for example shifting item orders to the previous and the next IGO is, that we increase $y^i_{s_{k+1}}$ and $y^i_{s_k}$ both by $\sum_{s \in [s_k+1,s_{k+1}-1]} y^i_s$. Of course, because we are shifting to the previous and the next, $y^i_{s_k} := \sum_{s=[s_{k-1}+1,s_{k+1}-1]} y^i_s$. In this process though, we are not changing $y^i_s$ for $s \notin \mathcal{T}$. The timesteps $s \notin \mathcal{T}$ will not, however be in instance 1 while the timesteps in $\mathcal{T}$ will not be in instance 2.
\paragraph{Feasible solutions for each instance:} As before, the optimal solution restricted to the timesteps and demand points in Instance 2 is a feasible solution for Instance 2. For Instance 1, we use the same approach as in section 3 to construct $\overline{y}^i_s$ variables by shifting the weight of every $y_s^i$ variables to the previous and next initial general order. 

We do something a little more parsimonious for the $x$ variables. For each $s \in \mathcal{T}, s \leq t$, let $q(s)$ be the previous IGO (if $s$ is the first timestep, $q(s)=0$ since the time horizon begins with $T=1$). We construct $\overline{x}^{it}_s$ by shifting any service between two IGOs to the later of the two when possible. Otherwise we just shift it to the earlier IGO.

More formally, $\overline{x}^{it}_s=\sum_{s'=q(s)+1}^s x^{it}_s$ unless $s$ is the last IGO before $t$ in which case, $\overline{x}^{it}_s = \sum_{s'=q(s)+1}^t x^{it}_s$.

\subsection{Pipage Rounding for Instance 2}
In this section, we will aim to prove Theorem \ref{thm: pipage_general}. We will assume that we are given an instance and an LP solution $(x,y,r)$ such that the desired partition of the time horizon into batches exists. 

Our algorithm is almost the same as before; we will consider groups of non-integral variables and adjust until at least one of them is integral, while maintaining feasibility and not increasing cost. At the end, we will need to round up at most $2C$ general orders, $2C$ item orders and $C$ individual demand point service variables which is why our solution might cost slightly more than $\LP^{\sol}$. We only need to round up $C$ individual demand point service variables and not $2C$ because the ``total number of orders is at most $1$" type restrictions are no longer necessary at the demand point level.

\begin{itemize}
\item[0.] \textbf{Trimming unnecessary orders:} The aim of this step is to make sure that there are no variables that are larger than they need to be. 

Recall that, by assumption, if $B(d)$ is the set of timesteps in the batch containing the deadline of $d \in D_2$, then $x^d_s=0$ if $s \notin B(d)$. Moreover, since $\sum_{s \in B} y_s \leq 1$ for each batch $B$, $\sum_{s \leq t} x^{it}_s \leq 1$ for every $(i,t) \in D_2$. Therefore, we can assume $r_{it}=1-\sum_{s \in B(i,t)} x^{it}_s=1-\sum_{s \leq t} x^{it}_s$.

For each item $i$ and timestep $s$ let $D^{is}_2$ be the set of demand points $(i,t)$ with non-zero $x^{it}_s$.  If possible, decrease $y^{i}_s$ so that $y^i_s= \max_{(i,t) \in D^{is}_2} x^{it}_s$. After doing this for each $i,s$ for each timestep $s$, if possible, decrease $y_s$ so that $y_s=\max_{i \in [N]} y^i_s$. Clearly, these operations are feasible and do not increase cost. We call a timestep $s$ ``trimmed" if none of these operations would actually result in a variable corresponding to $s$ decreasing.

\item[1.] \textbf{Splitting Timesteps:}  We now show that WLOG, $x^{d}_s \in \{0, y_s\}$ for each timestep $s$ in this instance (instance 2) by splitting these timesteps. The figure below gives a clear visual explanation; however, we also provide a more algebraic explanation immediately after it:

\begin{figure}[h]
\includegraphics[width=16cm]
{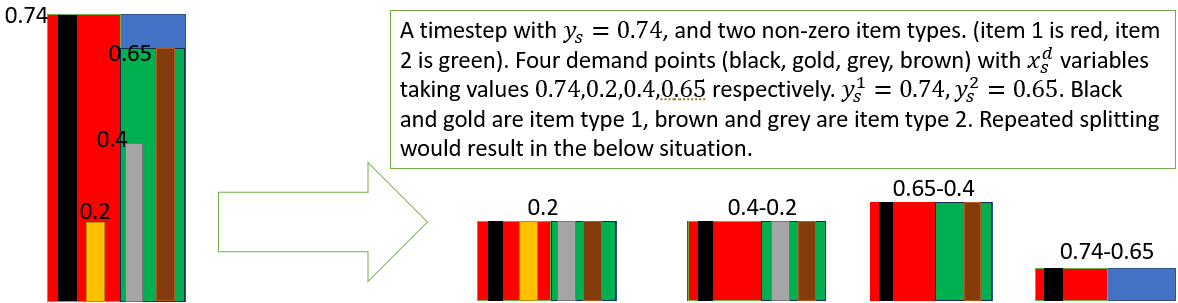}
\centering
\end{figure}

Pick any timestep $s$ with some $d$ such that $x^d_s \in (0,y_s)$ find the demand $d_{\low}$ with the lowest $x^{d_{\low}}_s$.  We now split this timestep into $s_1,s_2$. For all demand points $d$ $H^d_{s_1}=H^d_{s_2}=H^d_s$ (if the demands' deadline is at or after $s$). Essentially, $s_1$ will have every positive variable corresponding to $s$ take value $x^{d_{\low}}_s$ and $s_2$ will have every positive variable in $s$ decreased by the same amount (This implies $x^{d_{\low}}_{s_2}=0$). 

If there exists some $d$ such that $x^d_{s_2} \in (0,y_{s_2})$, then repeat this process with $s_2$. If not then, pick any timestep $s$ with some $d$ such that $x^d_s \in (0,y_s)$.  If none exits, then proceed to Step 2. 

Since at every timestep $s$ there are at most $|D_2|$ many demands with $x^d_s>0$ and every such split reduces the number of such demands at the resultant $s_2$, the number of such splits that correspond to the original timestep $s$ and its' descendents is $|D_2|$. Therefore, we have to complete this process at most $|D_2|T$ times.

See the figure below for an illustration of the outcome of this process. A more algebraic description is provided in the following paragraph.

More concretely, for all demand points $d$ with $x^{d}_s>0$, $x^d_{s_1}=x^{d_{\low}}_{s}$ and $x^d_{s_2}=x^d_s-x^{d_{\low}}_{s}$. Notice that if for any item type $j$ $y^j_s<x^{d_{\low}}_s$, $y^j_s = 0$ because of trimming. Therefore, for all $i$ such that $y^i_s > 0$, we set $y^i_{s_1}=x^{d_{\low}}_s$ and $y^i_{s_2}=y^i_s-x^{d_{\low}}_s$. We will also set $y_{s_1}=x^{d_{\low}}_{s}$ and $y_{s_2}=y_s-y_{s_1}$. 

Notice that this splitting operation does not increase cost and results in a feasible solution serving every demand to the same extent as before. Moreover, notice that both the new timesteps are trimmed and we have $x^d_{s_1} \in \{0,y_{s_1}\}$. Applying this process until all the demand points in the resulting $s_2$ are serviced to exactly the same extent as each other gives us the desired candidate order structure.

Each timestep in the original instance is split at most $|D_2|$ times in this way ($s_1$ will never be split further and $L(s_2)$ has one fewer demand point than $L(s)$). Therefore, the size of the instance is still polynomial in the input. For convenience, we will still let $T$ denote the number of timesteps. 
\item[2.] \textbf{Rounding Candidate Orders via Piping:} For each demand point $d \in D_2$, let $s(d)$ be the earliest timestep such that $x^d_s>0$ and $t(d)$ be the deadline of $d$. Because of our assumptions on the structure of the LP solution, $I(d):=[s(d),t(d)]$ lies in a single batch. For each timestep $s$ with $y_s>0$, we define a candidate (replenishment) order, consisting of all item types with $y^i_s >0$. Once again, if $y^i_s > 0$, then $y^i_s = y_s$. A \textit{non-integral} candidate order is one where $0 < y_s < 1$. Increasing a non-integral order by $\delta$ refers to increasing $y_s$ and each positive $y_s^i$ by $\delta$. An \textit{integral} order is one where $y_s \in \{0,1\}$. We loop over steps $a) \text{ and }b)$ until only $2C$ non-integral candidate orders remain. Then we round the $2C$ corresponding $y_s$ variables up to $1$ and proceed to step $3.$
\end{itemize}

\begin{itemize}
\item[(a)] We construct a linear system of equations as follows. The components of our solution vector will specify a perturbation each non-integral candidate order so that we can move towards an integral solution without increasing the total cost. 
For each non-integral candidate order $o$, there will be a corresponding variable $\delta_o$ in the linear system. A solution to the system will specify a feasible perturbation to the current solution, and it will be possible to move towards an integral solution without increasing the total cost.
Let $w^c_o$ be the total weight of demand points $d$ in $D_2$ of color $c$ such that candidate order $o$ would serve $d$ if it were placed. The first $C$ equations are $\sum_{o} w^c_o\delta_o = 0$ for each color $c$. For every batch $B$ with multiple non-integral candidate orders, we also add the constraints $\sum_{o \in B} \delta_o = 0$. Let $G$ be the set of such batches. Critically, for each batch $B$ we have that $\sum_{s \in B} y_s \leq 1$ and hence $|D_2|-\sum_{d \in D_2} r_d = \sum_{o} w^c_o y_o$ where $y_o$ is the value of $y_s$ at the timestep $s$ corresponding to order $o$.
    \item[(b)] We find a non-zero solution to the above system of equations. Let $d_{ord}$ be the set of candidate orders that serve demand $d$. If, for every candidate order $o$, increasing the corresponding $y_s$ and each $y^i_s$ such that $y^i_s=y_s$ as well as each $x^{it}_s$ such that $x^{it}_s=y_s$ at rate $\delta_o$ is a change that does not lead to a net increase in cost, then do so and simultaneously decrease $r_d$ at a rate $\sum_{o \in d_{ord}} \delta_o$ until some variable becomes integral. Else, increase the corresponding $y_s$ and all $y^i_s$ such that $y^i_s=y_s$ at rate $-\delta_o$ and simultaneously increase $r_d$ at a rate $\sum_{o \in d_{ord}} \delta_o$ until some variable becomes integral. The number of variables and constraints in the linear system guarantees the following lemma.
\end{itemize}

    \begin{lemma}
    \label{colorful_candidate_orders}
Whenever there are $2C+1$ or more non-integral candidate orders, there is a non-zero solution to the linear system.
    \end{lemma}
\begin{proof}
Let $G$ be the set of batches with multiple fractional candidate orders. The system has at least $2G$ variables and at least $G+C$ constraints. This implies that whenever $G > C$, there are more variables than constraints. Since all constraints have a 0 on the RHS, the system always has a feasible solution (everything 0). Therefore, if there are more variables than constraints, the system must have infinitely many solutions (in particular, infinitely many solutions with some $o$ such that $\delta_o \neq 0$).

If $G \leq C$, then whenever there are $G+C+1 \leq 2C+1$ or more fractional candidate orders, the system has a non-zero solution. 

Therefore, $2C+1$ or more fractional candidate orders implies the existence of a non-zero solution to the linear system.
\end{proof}
\begin{itemize}
    \item[3.] \textbf{Rounding item orders} We now round the $y^i_s$ in a similar manner to the $y_s$ earlier. Loosely speaking, one can think of each (item,batch) pair as corresponding to a batch in step $2$ and whenever there are $2C+1$ or more non-integral $y^i_s$ variables, we can perform steps $(a)$ and $(b)$ below. Once there are $2C$ or fewer non-integral $y^i_s$ variables we round them up.
    \item[(a)] For each non-integral $y^i_s$, let $w^c_{is}$ be the total weight of demand points in this instance of item $i$ with $x^{it}_s>0$. Once again, we solve a linear system of equations with variables $\delta^i_s$ corresponding to these non-integral $y^i_s$. The first $C$ constraints are $\sum_{i,s} w^c_{is}\delta^i_s = 0$. For every batch $B$ and item type $i$ with multiple non-integral $y^i_s$, we add a constraint $\sum_{s \in B} \delta^i_s = 0$. For similar reasons to before, we get the following lemma guaranteeing that we can find such a solution:
\end{itemize}    
\begin{lemma}
\label{colorful_item_orders_pipage}
Whenever there are $2C+1$ or more non-integral $y^i_s$ variables, there is a non-zero solution to the linear system. 
\end{lemma}
\begin{proof}
Identical to the previous lemma, except $G$ will now be the set of (item,batch) pairs $(i,B)$ with multiple $y^i_s>0$ such that $s \in B$.
\end{proof}
\begin{itemize}
    \item[(b)] Similarly to step $2(b)$, we change $y^i_s$ variables at rate $\delta^i_s$ until some $y^i_s$ becomes integral. We also change $x^{it}_s$ variables at the same rate and $r_{it}$ variables are decreased so that $\sum_s x^{it}_s + r_{it}=1$.
    \item[4.] \textbf{Rounding demand point service variables:} Now that all the item ordering and general ordering variables are still fixed, we can apply the same method as we did when dealing with ``Small solutions: Case 2" and this increases costs by at most $CH^M_{\max}$ by Lemma \ref{fixed_replenishment_schedule_lemma}
\end{itemize}

Clearly, the algorithm above returns an integral solution. All that remains is to bound the cost and verify feasibility.

\begin{lemma}
If $\LP^2$ is the cost of the original solution, then the cost of the rounded solution is at most $\LP^2+2CK_0+2CK^{M}_{\max}+CH^{M}_{\max} \leq (1+\frac{\epsilon}{8})\LP^2$
\label{pipage_cost_general}
\end{lemma}
\begin{proof}
The cost of the solution only increases when we round up $2C$ variables in steps $2.$ and $3.$ and when we round up $C$ variables in step $4.$ (by Lemmas \ref{colorful_candidate_orders}, \ref{colorful_item_orders_pipage} and \ref{fixed_replenishment_schedule_lemma} respectively). The extra cost in step $2.$ is at most $2K_0$, the extra cost in step $3.$ is at most $2K^{M}_{\max}$ and the extra cost in step $4.$ is at most $CH^{M}_{\max}$ (by Lemma \ref{fixed_replenishment_schedule_lemma}). This gives us the bound on the cost of the rounded solution.  $\LP^2+2CK_0+2CK^{M}_{\max}+CH^{M}_{\max} \leq (1+\frac{\epsilon}{8})\LP^2$ follows from our choice of $M$.
\end{proof}

\begin{lemma}
\label{pipage_feasibility_general}
The algorithm always maintains a feasible solution.    
\end{lemma}

\begin{proof}
Notice that by construction, the algorithm always maintains that $x^{it}_s \leq y^i_s \leq y_s$ and that we always adjust $r_{it}$ to ensure that $\sum_{s \leq t} x^{it}_s+r_{it} = 1$. 

Moreover, the choice of constraints that the perturbations in each iteration of steps 2 and 3 must satisfy ensures that $\sum_{s \leq t} x^{it}_s \leq 1$ and that
$\sum_{d \in D_2} w^c_d \sum_{s \leq t} x^{it}_s = \sum_{d \in D_2} w^c_d  - \sum_{d \in D_2} w^c_d r_d$ does not change until the end of step 3. This ensures that until the end of step 3 the rejection limit constraint is satisfied and that $r_{it} = 1- \sum_{s \leq t} x^{it}_s \geq 0$. We know that after step 4. we still satisfy the rejection limits and $r_{it}\geq 0$ by Lemma $\ref{fixed_replenishment_schedule_lemma}$

Lastly, $x^{it}_s,y^i_s,y_s \geq 0$ by construction (the moment a candidate order/item order variable becomes integral, it never again gets perturbed by the algorithm). 

Therefore, the integral solution at the end of step 4 satisfies the constraints. 

\end{proof}

Lemma \ref{pipage_feasibility_general} and Lemma \ref{pipage_cost_general} imply that we obtain a feasible solution of cost at most $(1+\frac{\epsilon}{8})\LP^2$ and by construction, our solution must be integral. This is essentially what Theorem \ref{thm: pipage_general} stated.
\subsection{Iterative Rounding for Instance 1}
This section will be dedicated to proving Theorem \ref{thm:iterative_general}, which essentially states that given a time horizon with general orders placed at every timestep and an LP solution $\overline{x},\overline{y},\overline{r}$ with item ordering cost $\LP^1_{\itm}$ and holding cost $\LP^1_{\hold}$, we can obtain an integral solution of cost $(1+\frac{\epsilon}{8})\LP^1_{\itm} + \LP^1_{\hold}$. Moreover, our integral solution $(x,y,r)$ will have $r_d=0$ if $\overline{r}_d=0$ and $r_d=1$ if $\overline{r}_d=1$. In future sections, we will refer to the input and output of this solution using the bar (i.e. as $\overline{x},\overline{y},\overline{r}$), however, in this section, to simplify notation, we refer to them as $(x,y,r)$. Like before, we will rely on a critical structural feature of an extreme point optimal solution. Because there are multiple colors and general holding costs, instead of there being at most $1$ multibatch, there will be at most $C+1$ multibatches. The rounding process (add extra interval constraints and round up a $C-$dependent constant number of item orders) will be similar to before.

Now we describe some key definitions and properties that will be critical to our algorithm.

\subsubsection{Definitions and Facts}
We state the LP formulation this section will be using. (Note that $y_s=1$ for all timesteps in any input instance for the iterative rounding algorithm, therefore, we do not include it. )

\begin{empheq}{align*}
\text{minimize} \quad & \sum_{i=1}^N \sum_{s = 1}^T y^i_s K_i + \sum_{(i,t) \in D} \sum_{s=1}^T H_s^{it} x_s^{it} + \sum_{(i,t) \in D} p_{it}r_{it}\\
\text{subject to }\quad & r_{it} + \sum_{s = 1}^t x_s^{it} \geq 1 \quad \text{for } (i,t) \in D \hspace{1.3in} \text{(Reject or service every demand)}\\
\quad & x_s^{it} \leq y_s^i \quad \text{for all } (i,t) \in D, s \in [T] \hspace{1.1in}\text{(Servicing requires an item order)}\\
\quad & \sum_{(i,t)\in D} w^c_{it} r_{it} \leq R^1_c \quad\text{for all } c \in [C] \hspace{1.2in} \text{(Rejection bound for each color)}\\
\quad & y_s,y_s^i,x^{it}_s,r_{it} \in \{0,1\} \quad \text{for all } s \in [T], i \in [N], (i,t) \in D 
\end{empheq}

For a given LP solution $(x,y,r)$, a \textit{multibatch} is an item, interval pair $(i,[s_1,s_2])$ such that the following 3 properties are satisfied:

\begin{enumerate}
    \item $y^i_{s_1}>0,y^i_{s_2} > 0$
    \item For all $s \in [s_1,s_2], y^i_s < 1$.
    \item $(i,[s_1,s_2])$ is maximal in the sense that if $i,[s_0,s_3]$ satisfies the first 2 properties and $[s_1,s_2] \subset [s_0,s_3]$ then $s_1=s_0,s_2=s_3$.
\end{enumerate}

Note that for every multibatch $i,[s_1,s_2]$, the smallest $s > s_2$ such that $y^i_s > 0$ will satisfy $y^i_s=1$(if $s_2$ is not the last timestep $s$ with positive $y^i_s$). Similarly, the largest $s < s_1$ such that $y^i_s>0$ will satisfy $y^i_s=1$ (if $s_1$ is not the first timestep $s$ with positive $y^i_s$). Let $s^{i[s_1,s_2]}_{\prev}$ denote this largest $s<s_1$ such that $y^i_s=1$. It is possible that some demand points are partially satisfied within $[s_1,s_2]$ and partially satisfied at $s^{i[s_1,s_2]}_{\prev}$. 


\textbf{Fact:} We have a disjoint collection of multibatches such that for every fractional $y^i_s$ there is a multibatch $(i,I)$ such that $s \in I$.

\textbf{Leanness:} We define the concept of a \textit{lean} solution as follows. For every demand point $(i,t)$, let $s(it)$ denote the earliest timestep $s$ at which $x^{it}_s>0$. We say that a solution is lean if for every $s \in [s(it)+1,t], x^{it}_s = y^i_s$ and $\sum_{s \leq t} x^{it}_s + r_{it}=1$ It should be apparent that we can convert any solution into a lean solution without increasing cost by monotonicity of the holding costs.

\textbf{Tight and semitight demand points:} We call a demand point $(i,t)$ \textit{tight} if $x^{it}_s =y^i_s$ for all $s \in [s(it),t]$, $\sum_{s \leq t} x^{it}_s=1$ and $\max_{s} x^{it}_s < 1$. Let $D_i^{\tight}$ be the set of tight demand points of item type $i$. Note that for tight demand points $(i,t)$ $\sum_{s=s(it)}^t y^i_s = \sum_{s \leq t} x^{it}_s=1$. We call a demand point \textit{semitight} if it satisfies all the above restrictions with one exception: $x^{it}_{s(it)}<y^i_s$. We let $D^i_{\semitight}$ denote the set of semitight demand points. Note that $D^i_{\tight}$ and $D^i_{\semitight}$ are disjoint.

\textbf{Tight intervals:} We define $V_i^{\tight}$ to be the set of intervals $[s_1,s_2]$ such that $\sum_{s=s_1}^{s_2} y^i_s = 1$. Recall that $V_i$ will be the set of intervals $[s_1,s_2]$ such that $\sum_{s=s_1}^{s_2} y^i_s \geq 1$ and we will add constraints corresponding to these intervals to the LP before re-solving.

Notice that if a solution is lean, then every demand point $(i,t)$ such that $r_{it}=0$ is either tight, or semitight or has $x^{it}_{s(it)}=1$ (in which case, it is already rounded and we don't need to worry about it). Moreover, because of the perturbation to the holding costs, the leanness property must be satisfied by an optimal solution. We formalize this in Lemma \ref{iterative_general_leanness_optimality} after describing our algorithm.

\subsubsection{At Most $C+1$ Multibatches}
In addition to the interval constraints that we saw in section 3, we will also need to add an extra constraint stipulating that upon re-solving the LP, the item ordering cost does not increase (In the deadlines case with no penalties, this constraint was unnecessary- the only type of cost that wasn't already fixed was the item ordering cost).

By definition, a feasible solution $x,y,r$ is an extreme point if there exists no vector $\delta$ such that $(x,y,r)+\delta$ and $(x,y,r)-\delta$ are both feasible. We will for each multibatch, find a $\delta$ such that $(x,y,r)+\delta$ and $(x,y,r)-\delta$ satisfy all constraints except the rejection limits and the ``item ordering cost doesn't increase" constraint. This is the essence of the following lemma:

\begin{lemma}
For each multibatch $(i,I)$ of a lean solution $(x,y,r)$, there exists a vector $\delta=\delta^{i't}_s,\delta^{i'}_s,\delta_{i't}$ (with entries corresponding to each variable in the LP) such that (1) For all integral variables, the corresponding entry of $\delta$ is $0$; (2) If $j \neq i$ or $s \notin I$, $\delta^j_s=0$. (3) $(x,y,r)+\delta$ still satisfies all constraints including the interval constraints we add during the iterative rounding process except the colorwise rejection limit constraints and the total item order cost constraint (4) There exists a $y^i_s$ such that $s \in I$ and the corresponding $\delta$ is non-zero (5) the above properties hold for $-\delta$ as well.
\label{iterative_feasible_direction_general}
\end{lemma}

\begin{proof}
    The proof of this lemma is lengthy; by far the lengthiest single lemma-proof in the entire article. Roughly speaking, we will first construct a $\delta$ vector by carefully choosing timesteps to perturb. Then we will prove that this perturbation satisfies the desired properties.
    Let $\overline{V}$ be the set of intervals $[s,t]$ entirely within the multibatch such that either $[s,t] \in V^i_{\tight}$ or there exists some demand point $(i,t) \in D^i_{\tight}$ such that $s(i,t)=s$. Let $\kappa$ be an arbitrarily small positive value. 
    
    Define an increasing sequence of timesteps as follows: $t_1$ is the first timestep in the multibatch with positive $y^i_s$. Let $t_2$ be the last timestep $s \in I$ such that $1>y^i_s>0$ and every interval in $\overline{V}$ that contains $t_1$ also contains $s$ (if any such intervals exist; if not the sequence just has one element). Given $t_k,t_{k-1}$, we define $t_{k+1}$ to be the last timestep $s \in I$ such that $1>y^i_s>0$ and every interval in $\overline{V}$ that contains $t_k$ but not $t_{k-1}$ also contains $t_{k+1}$ (if no such interval exists, the sequence terminates at $t_k$).

We now construct the $\delta$ vector. In this construction, $\tau$ is an extremely small number. Every constraint (other than the colorwise rejection limits and item order cost limit constraints) in the ILP contains at most $3|D|NT$ variables, each with a coefficient in $\{0,1,-1\}$. We choose $\tau = \frac{f}{9(|D|NT)^2}$ where $f$ is the smallest non-zero slack in any inequality constraint at the solution $(x,y,r)$
\begin{itemize}
    \item[0.] Initially, $\delta$ is just the 0 vector.
    \item[1.] For all $s \in \{t_1,t_3,t_5...\}$ the $\delta^i_s$ is changed to $+\tau$ and for all $s \in \{t_2,t_4,t_6...\}$ the $\delta^i_s$ is changed to $-\tau$. For all $s \in \{t_1,t_2,t_3...\}$ and demand points $(i,t)$, if $x^{it}_s=y^i_s$ then $\delta^{it}_s$ will be the same as $\delta^i_s$. 
    \item[2.] If $(i,t) \in D^i_{\semitight}$, then $\delta^{it}_{s(it)} := -\sum_{s \in [s(it)+1,t]} \delta^{it}_s$. Notice that $\delta^{it}_s=\delta^i_s$ for $s>s(it)$ because the solution is lean. 
    \item[3.] For all demand points $(i,t)$ such that for some $s \in I$, $x^{it}_s>0$ and $(i,t)$ is neither tight nor semitight, we note that $1>r_{it}>0$. Therefore, we choose the $\delta_{it}$ corresponding to $r_{it}$ to be $-\sum_{s \in [s(it),t]} \delta^{it}_s$.       
\end{itemize}
Now we must prove that the constructed $\delta$ satisfies the above properties.

\textbf{Property 4:} By construction $\delta^i_{s_1} \neq 0$

\textbf{Property 2:} Notice that $t_1,t_2...$ are all members of the multibatch and the only $s$ with $\delta \neq 0$ corresponding to $y^i_s$ are in $t_1...t_k$. Therefore, if $s \notin$ $\{t_1,t_2...\} \subset [s_1,s_2]$, then the $\delta^i_s=0$.

\textbf{Property 1:} We prove this for each type of variable $(x,y,r)$ separately.

By construction, if $s \notin \{t_1,t_2,t_3...\}$, $\delta^i_s=0$. Because $t_1,t_2... \in [s_1,s_2]$, $0<y^i_{t_k}<1$. Therefore, if $y^i_s$ is integral, $\delta^i_s=0$.

The only way $\delta_{it} \neq 0$ if there exists some $s \in I$ such that $x^{it}_s > 0$ and $(i,t) \notin D^i_{\tight} \cup D^i_{\semitight}$. However, if our solution is lean, these imply that $0<\sum_{s \leq t} x^{it}_s<1$ and that $r_{it}=1-\sum_{s \leq t} x^{it}_s$ is non-integral. Therefore, if $r_{it}$ is integral, $\delta_{it}=0$

If $\delta^{it}_s \neq 0$ then by construction, $x^{it}_s > 0$. If $x^{it}_s = 1$ then $s \notin [s_1,s_2]$ and by the leanness property, since $\sum_{s' \leq t} x^{it}_{s'} = 1, \sum_{s' \in [s_1,s_2]} x^{it}_{s'} = 0$. This implies that $\delta^{it}_s$ could not have been changed from its initial value of $0$ in any of steps 1-3 of the construction of $\delta$. Therefore, if $x^{it}_s$ is integral, $\delta^{it}_s=0$.

\textbf{Property 3:} The only constraints we need to check are $(x,y,r)+\delta \geq 0, \sum_{s \leq t} (x^{it}_s + \delta^{it}_s) + r_{it}+\delta_{it} \geq 1$ and $x^{it}_s +\delta^{it}_s \leq y^i_s+\delta^{i}_s$ and the extra interval constraints. If a variable is strictly positive or $x^{it}_s<y^i_s$, we choose $\tau$ to be small enough that the corresponding non-negativity or $x^{it}_s+\delta^{it}_s \leq y^i_s+\delta^i_s$ constraint must be satisfied.

Therefore, we turn our attention to tight constraints. Previous properties already take care of tight non-negativity constraints.

\textbf{Demands cannot be served without an item order:} If $x^{it}_s=y^i_s$, by definition, $\delta^{it}_s=\delta^i_s$ and therefore, $x^{it}_s + \delta^{it}_s \leq y^i_s+\delta^i_s$.

\textbf{Extra interval constraints:} If there exists an interval $I$ such that the constraint $\sum_{s \in I} y^i_s \geq 1$ is tight, then notice that this interval must have been in $\overline{V}$. 

We now claim that every interval $I_v=[s_1,s_2]$ in $\overline{V}$ has either zero or two successive elements of $t_1,t_2...$ in it. Clearly, if $I_v$ has $t_k$ but no element of $\{t_1,t_2,t_3...\}$ preceding $t_k$, by construction of the sequence it must include $t_{k+1}$. Let $I'_v=[s_3,s_4] \in \overline{V}$ be the interval that induces the placement of $t_{k+2},$; the interval containing $t_{k+1}$ but not $t_k$ that had the earliest right endpoint. By definition of $I_v$ and $I'_{v}$ we know that $s_1 \leq t_k < s_3 \leq t_{k+1} \leq s_2$. 

Since $y^i_{t_k}>0$, $\sum_{s \in [s_3,s_2]} y^i_s < 1 = \sum_{s \in [s_1,s_2]} y^i_s$ we can conclude that $s_4>s_2$ and that $\sum_{s \in [s_2+1,s_4]} y^i_s>0$. Therefore, the last timestep $s$ in $[s_3,s_4]$ such that $y^i_s>0$ which is the timestep $t_{k+2}$ must be after $s_2$ and therefore, can't be in the interval $[s_1,s_2]$.

Remember though that the $\delta^i_s$ for $t_1,t_3,t_5...$ was $\tau$ and for $t_2,t_4,t_6...$ was $-\tau$. Therefore, if any interval $I \in \overline{V}$ contained two consecutive members of the sequence, $\sum_{s \in I} \delta^i_s=0$. Moreover, if $I$ didn't contain any members of the sequence, $\sum_{s \in I} \delta^i_s=0$. As we know that every $I \in \overline{V}$ contains either $0$ or $2$ members of the sequence, $\sum_{s \in I} (y^i_s+\delta^i_s)=\sum_{s \in I} y^i_s = 1 $.

This implies that $(x,y,r)+\delta$ satisfies the interval constraints. 

\textbf{Remark:} We can use a very similar argument (coupled with the fact that $\delta^{it}_s = \delta^i_s$ when $x^{it}_s=y^i_s$) to conclude that for tight demand points $(i,t)$, $\sum_{s \leq t} x^{it}_s +\delta^{it}_s= \sum_{s \leq t} x^{it}_s = 1$

\textbf{Demands must be served or rejected:} The previous remark explains how we can prove the result for tight demand points. For semitight demand points, we defined $\delta^{it}_{s(it)}=-\sum_{s\in [s(it)+1,t]} \delta^{it}_{s} \implies \sum_{s \in [s(it),t] }\delta^{it}_s = 0$ ensuring that $\sum_{s \leq t} x^{it}_s+ \delta^{it}_s= \sum_{s \leq t} x^{it}_s + \sum_{s \leq t} \delta^{it}_s=\sum_{s \leq t} x^{it}_s = 1$ implying that the constraint is satisfied (recall $r_{it}=0$ if $(i,t)$ is semitight).

Recall that if a demand point $(i,t)$ is neither tight nor semitight, either $r_{it}>0$ or $x^{it}_s=1$ for some $s$. We have already verified that in the latter case all components of $\delta$ corresponding to this demand point are zero. The only case that remains is $r_{it}>0$.  We only choose a non-zero $\delta$ corresponding to $r_{it}$ if $x^{it}_s>0$ for some $s \in I$ and $(i,t)$ is neither tight nor semitight. By the leanness property, there is no $s'$ such that $x^{it}_{s'}=1$ and therefore, $r_{it} \neq 0$. Since $x^{it}_s>0$, $r_{it} \neq 1$. Therefore, we only choose a non-zero $\delta$ for $r_{it}$ if $0<r_{it}<1$. In this case the choice of $\delta_{it}$ clearly ensures that the constraint is still satisfied.

\textbf{Property 5:} Notice that we could make exactly the same arguments for $-\delta$ instead of $\delta$ by choosing to set $\delta^i_{t_k} = -\tau$ for $t_k \in \{t_1,t_3,t_5...\}$ and $\delta^{i}_{t_k}=\tau$ for $t_k \in \{t_2,t_4,t_6...\}$
\end{proof}

Now we use this result to argue that at an extreme point optimal solution, there are at most $C+1$ multibatches.

\begin{lemma}
At an optimal extreme point solution to the LP (possibly with extra interval constraints, the item order cost limit constraint and constraints fixing all variables of certain item types), there are at most $C+1$ multibatches of the unfixed item types.
\label{iterative_general_multibatch_bound}
\end{lemma}

\begin{proof}
Let $\mathcal{B}$ be the collection of multibatches of the unfixed item types. We will prove later, in Lemma \ref{iterative_general_leanness_optimality} that an optimal solution must be lean. This allows us to invoke Lemma \ref{iterative_feasible_direction_general} giving us, for each multibatch $B$, a vector $\delta_B$ such that $(x,y,r)+\delta_B$ and $(x,y,r)-\delta_B$ each violate at most $C+1$ constraints (the $C$ rejection limit constraints and the item order limit constraint). 

We let $g^c_B:=\sum_{d \in D_1} w^c_d {\delta_B}_{d}$ and $g_B:=\sum_{s} {\delta_B}^{i(B)}_s$ where $i(B)$ is the item type of the multibatch $B$. Now we solve a linear system with a variable $\gamma_B$ corresponding to each multibatch and constraints $\sum_{B \in \mathcal{B}} g^c_B \gamma_B = 0$ for all colors $c$ as well as $\sum_{B \in \mathcal{B}} g_B \gamma_B = 0$.

When there are $C+2$ or more multibatches, there exists a non-zero solution $\gamma$ to this system of equations. Because LPs are convex, it follows that if we define $\delta = \frac{\sum_{B \in \mathcal{B}} \gamma_B \delta_B}{\sum_{B \in \mathcal{B}} |\gamma_B|}$, both $(x,y,r)+\delta$ and $(x,y,r) - \delta$ are feasible solutions, violating the assumption that $(x,y,r)$ was an extreme point.
\end{proof}

\textbf{Clarification:} In the special case we looked at earlier ($\rjrpd$), the only cost was the item ordering cost. Therefore, the constraint that the item ordering cost doesn't increase after re-solving was unnecessary, we already had this property. That's why we had only a single multibatch instead of two. 

\subsubsection{The Iterative Rounding Algorithm}
Now we describe the actual iterative rounding algorithm. We assume we are given a non-integral starting solution $(x,y,r)$
\begin{itemize}
    \item[0.] \textbf{Initialization:} We add constraints forcing all integral variables to keep the same value. Then we re-solve to find an extreme point optimal solution.
    \item[1.] \textbf{Finding an item type to round:}  Let $i$ be an item type with some $s$ such that $y^i_s$ is non-integral (if none exists, then our solution is integral and no rounding is required). We fix all variables not corresponding to item type $i$ and treat them as constants until every variable associated with item $i$ is integral. Let $\LP^i_{\itm}$ be the item ordering cost incurred by variables of item type $i$, $Q_i=\sum_{s: y^i_s<1} y^i_s$ and $\LP^i_{\hold}$ be the holding cost incurred by variables of item type $i$ before any iterations have been performed.
    \item[2.] \textbf{An iteration:} The high level idea in each iteration is that we want to split the largest multibatch of the unfixed item type $i$ into $2C+2$ smaller multibatches by rounding up some $y^i_s$ variables to $1$. We will also identify a set of intervals $\mathcal{J}$ and add constraints that $\sum_{s \in J} y^i_s \geq 1$ for $J \in \mathcal{J}$. We will also add a constraint ensuring that the item ordering cost after re-solving is no higher than immediately before re-solving. Because of Lemma \ref{iterative_general_multibatch_bound} we know that after re-solving and finding a new extreme point optimal solution, there will be at most $C+1$ multibatches (which by definition cannot contain any of the ``rounded up" timesteps). Therefore, at least $2C+2-(C+1)=C+1$ of these smaller multibatch intervals will have no fractional variables after re-solving the LP. An iteration comprises of the following sequence of steps:
    \item[(a)] \textbf{Choosing the largest multibatch:} Define the size of a multibatch $i,[s_1,s_2]$ to be $\sum_{s \in [s_1,s_2]} y^i_s$. Let $i,I := i,[s_1,s_2]$ refer to the largest multibatch and $Z$ be its size. If $Z < 16C+16$ then we proceed to the adjusted pipage rounding step. Otherwise, we perform an iteration. 
    \item[(b)] \textbf{Rounding up some variables:} 
    
    Let $Z^i_t=\sum_{s=s_1}^t y^i_s$. We let $Z=Z^i_{s_2}$ and construct a disjoint set of intervals, as follows. Let $t_1$ be the first timestep $t$ such that $Z^i_t \geq 1$. We define $I_1$ to be the interval. $[s_1,t_1]$. Given the disjoint intervals $I_1,I_2...I_{k-1}$ and sequence of right endpoints of these intervals $t_1,t_2...t_{k-1}$, we find the first timestep $t$ such that $Z^i_{t} - Z^i_{t_{k-1}} \geq 1$. Then we define the interval $I_k = [t_{k-1}+1,t_k]$. Because each such interval has $\sum_{s \in I_k} y^i_s \in [1,2)$, we can construct at least $\lfloor \frac{Z}{2} \rfloor \geq \frac{Z}{2}-1$ in this way. Let $A=\lfloor \frac{\frac{Z}{2}-1-(2C+1)}{2C+2} \rfloor=\lfloor \frac{Z}{4C+4} \rfloor - 1$. We round up $t_{A+1},t_{2A+2},t_{3A+3}...t_{(2C+1)(A+1)}$.

    \item[(c)]  \textbf{Adding extra constraints:} Let $\mathcal{J}$ be the set of intervals $J$ such that our current solution satisfies $\sum_{s \in J} y^i_s \geq 1$. For all $J \in \mathcal{J}$ we add the constraint the $\sum_{s \in J} y^i_s \geq 1$. Also, let $U=\sum_{s} y^i_s$, the total of all $y^i_s$ variables corresponding to item $i$. We add the constraint $\sum_{s} y^i_s \leq U$. We re-solve the LP and then begin the next iteration.
\end{itemize}
\begin{itemize}     
\item[3.] \textbf{Adjusted Pipage Rounding:} 
We know that each multibatch has size at most $16C+16$ and that there are at most $C+1$ multibatches. For each multibatch $B$, let $s_1(B),s_2(B)$ be the left and right endpoints of the corresponding interval respectively. We let $Z^i_{Bt}=\sum_{s=s_1(B)}^t y^i_s$. We let $t^B_k$ be the earliest timestep such that $Z^i_{Bs_2(B)} \geq k$ for all $k \in \{0,1,2,...\lfloor Z^i_{Bt} \rfloor\}$. For all $B$ and all $k\in \{0,1,,2,...\}$, we round up $y^i_{t^B_k}$ to $1$ if $t^B_k$ is defined. We will use $T^i$ to denote this set of all $t^B_k$ across all multibatches $B$. 

Now we let $o(it)$ denote the last $t^B_k$ before $t$ for each demand point $(i,t)$. Set $x^{it}_{o(it)}= \sum_{s \leq o(it)} x^{it}_s$ and for all $s<o(it)$ we set $x^{it}_s=0$. Notice that now for any multibatch $B$, the total of $y^i_s$ variables between any two successive $t^B_k$ is at most $1$. Each demand point only has positive $x^{it}_s$ if $s \geq o(it)$. 

\item[(a)] \textbf{Splitting timesteps:} We split timesteps with non-integral $y^i_s$ so that $x^{it}_s \in \{0,y^i_s\}$. Note that because these split timesteps correspond to the same timestep in the original instance, we can keep $y_{s'}=1$ for each post-split timestep $s'$ without increasing the cost of the returned instance.

\item[(b)] Define the candidate order corresponding to each timestep $s$ such that $s \notin T^i$ to be the set of demand points such that $x^{it}_s=y^i_s$. Note that by construction, $y^i_s=1 \iff s \in T^i$ and therefore $o(it) \in T^i$. Also note that between any two successive members of $T^i$, $v,w$ $\sum_{s \in [v+1,w-1]} y^i_s \leq 1$ which is a structural feature we exploited heavily in the pipage rounding algorithm ($[v+1,w-1]$ is analgous to a batch). The difference here is there are some demand points with deadlines in a batch $[v+1,w-1]$ which are partially served at the timestep immediately before the batch ($v$). 

Therefore, we define a linear system of equations with variables $\delta_o$, one for each candidate order with non-integral $y^i_s$ and variables $\gamma_{it}$ for each strictly positive $x^{it}_{o(it)}$. Let $\mathcal{A}$ denote the set of such demand points. For each $o$, let $s(o)$ be the timestep corresponding to $o$ and we let $w^c_o$ denote the total weight of demand points in the candidate order $o$. Our linear system will have constraints of the form $\sum_{o} w^c_o \delta^i_{s(o)} + \sum_{(i,t) \in \mathcal{A}} w^c_{it}\gamma_{it}=0$ for each color $c$ and for every batch $B$ with multiple non-integral candidate orders, we also add the constraints $\sum_{o \in B} \delta_o = 0$. For every demand point in $\mathcal{A}$ we also add a constraint $\sum_{o:}$

Like before, we can find a non-zero solution to this system of equations whenever there are at least $2C+1$ non-integral candidate orders (this system has an extra variable and an extra constraint for each element of $\mathcal{A}$ but $\delta = 0, \gamma=0$ is still a valid solution so the proof is identical to before). If such a non-zero solution exists, like before we perturb the solution at rates $\delta,\gamma$ until a variable becomes integral and repeat the process.

\item[(c)] \textbf{At most 2C non-integral item orders:} We round up the remaining $2C$ non-integral item orders. Then we perform the same procedure as in Lemma \ref{fixed_replenishment_schedule_lemma} to obtain integral $x$ variables for all demands of item type $i$.

\item[(d)] Go back to step 1. and choose a different item type to round.
\end{itemize}

\subsubsection{Analysis}
First we prove that an optimal solution at any iteration must be lean, which was useful in proving that there must be at most $C+1$ multibatches in an extreme point optimal solution.

\begin{lemma}
At any iteration, an optimal LP solution must be lean.
\label{iterative_general_leanness_optimality}
\end{lemma}

\begin{proof}
Suppose there exists an optimal LP solution $(x,y,r)$ which is not lean. That means either:
\begin{enumerate}
    \item There exists some demand point $(i,t)$ and some $s \in [s(i,t)+1,t]$ such that $y^i_s>x^{it}_s$. Let $\beta=min(y^i_s-x^{it}_s,x^{it}_{s(it)})$. Decrease $x^{it}_{s(it)}$ by $\beta$ and increase $x^{it}_s$ by $\beta$. This clearly does not violate any constraints and because of the $\frac{\kappa \epsilon (t-s)}{|D|T}$ perturbation to the holding costs right before we solved the LP in step (a) in the beginning of this section, we know this leads to a strict improvement to the objective value of the solution, violating the assumption that $(x,y,r)$ is optimal.
    
    OR
    
    \item There exists a demand point such that $\sum_{s \leq t} x^{it}_s + r_{it}>1$. In this case we take any positive $x^{it}_s$ decrease it slightly while maintaining $\sum_{s \leq t} x^{it}_s + r_{it} \geq 1$ and we would obtain a strictly better objective value. Moreover, this reduced solution is also clearly feasible violating the assumption that $(x,y,r)$ is optimal.
\end{enumerate}
In either case, we have a contradiction. Therefore, an optimal solution must be lean.
\end{proof}

Now we seek to bound the extra cost accrued during the rounding process. When rounding any single extra item type, we incurred extra costs due to rounding up variables before re-solving in step 2 and during the adjusted pipage rounding process in step 3. The following sequence of lemmas will help us establish these bounds (and show that the number of iterations this algorithm needs is logarithmic in the input size). 

\begin{lemma}
During the algorithm, when re-solving and only changing variables for a single item $i$, the new solution will have at most $C+1$ multibatches of item $i$.
\end{lemma}

\begin{proof}
Similar to Lemma \ref{iterative_general_multibatch_bound}.
\end{proof}

\begin{lemma}
The increase in cost from one iteration to the next is at most $(2C+1)K_i$.
\end{lemma}

\begin{proof}
In each iteration we round up at most $2C+1$ variables in the largest multibatch to $1$. All other constraints that are added to the LP were already satisfied by the current solution. Therefore the cost increases by at most $(2C+1)K_i$
\end{proof}

\begin{lemma}
The increase in cost due to adjusted pipage rounding for a single item type $i$ is at most $(16(C+1)^2+3C+1)K_i+CH^M_{\max} \leq (20(C+1)^2)K_i + CH^M_{\max}$
\end{lemma}

\begin{proof}
Each multibatch has total item ordering cost at most $16(C+1)$. When rounding up, we place at most $\lceil Z^i_{Bs_2(B)} \rceil \leq 16(C+1)+1$ orders of item $i$. There are at most $C+1$ multibatches so the cost of doing this is at most $(16(C+1)^2+C+1)K_i$. The pipage rounding algorithm incurs extra cost of at most $2CK_i+CH^M_{\max}$ due to rounding up some item orders and some demand point service variables. This gives us a bound of $(16(C+1)^2+3C+1)K_i + CH^M_{\max}$ as desired.
\end{proof}

\begin{lemma}
In each iteration of the iterative rounding procedure, the total cost of fractional $y^i_s$ variables associated with item type $i$ decreases by $(1-\frac{1}{8(C+1)})$. Consequently we need at most $\log_{\frac{8C+8}{8C+7}} (Q_{\init})$ many iterations where $Q_{\init} = \sum_{s} y^i_s$ before any iterations took place.
\end{lemma}

\begin{proof}
Let $Z$ be the size of the largest multibatch (before rounding up variables), we know that before that the $2C+1$ rounded up variables split the multibatch into $2C+2$ smaller multibatches, each with $\lfloor \frac{Z}{4C+4} \rfloor - 1$ many disjoint interval constraints. 

Moreover, at least $C+1$ of these will not contain any multibatch upon re-solving (and consequently, won't contain any non-integral $y^i_s$). Before re-solving, none of these intervals had an integral $y^i_s$. Because the solution before re-solving was feasible and we added the constraint that the $\sum_{s} y^i_s$ does not increase, the total of all non-integral $y^i_s$ variables $\sum_{s: y^i_s<1} y^i_s$ must decrease by at least $(C+1)(\lfloor \frac{Z}{4C+4} \rfloor - 1) \geq \frac{Z}{4} - (2C+2) \geq \frac{Z}{8}$ because $Z \geq 16C+16$.

There are at most $C+1$ multibatches of item type $i$. Therefore, $Z \geq \frac{\sum_{s: y^i_s<1} y^i_s}{C+1}$ and $\frac{Z}{8} \geq \frac{\sum_{s: y^i_s<1} y^i_s}{8(C+1)}$. This implies that after re-solving, $\sum_{s: y^i_s<1} y^i_s$ will decrease by a factor of $(1-\frac{1}{8(C+1)})$ This proves the first part of the lemma.

The second part: Let $Q^{\fracc}_{\init} := \sum_{s:y^i_s < 1} y^i_s$ before any iterations were performed. Let $Q^{\fracc}_k$ be the value of $\sum_{s: y^i_s < 1} y^i_s$ after $k$ iterations. We know that $Q^{\fracc}_k \leq (1-\frac{1}{8(C+1)})^k Q^{\fracc}_{\init} \leq (1-\frac{1}{8(C+1)})^k Q_{\init}$.

We stop iterating when the largest multibatch has size at most $16C+16$. But $(1-\frac{1}{8(C+1)})^k Q_{\init} \leq 16C+16$ if $k \geq - \log_{(1-\frac{1}{8(C+1)})} {\frac{Q_{\init}}{16C+16}}=\log_{(\frac{8C+8}{8C+7})} {\frac{Q_{\init}}{16C+16}}$

Therefore, we need at most $\lceil \log_{(\frac{8C+8}{8C+7})} {\frac{Q_{\init}}{16C+16}} \rceil \leq \log_{(\frac{8C+8}{8C+7})} {\frac{Q_{\init}}{16C+16}}+1 \leq \log_{\frac{8C+8}{8C+7}} (Q_{\init})$ many iterations for item $i$.




\end{proof}

\begin{corollary}
The total increase in cost due to rounding for item $i$ is at most $\frac{\epsilon\opt}{8 (C+1)}$
\label{iterative_general_cost_opt_itemwise_bound}
\end{corollary}

\begin{proof}
Combining the previous lemmas bounding cost accrued during iterations and adjusted pipage rounding, we know that the increase in cost due to rounding item $i$ is at most $(20(C+1)^2+ (2C+1)\log_{\frac{8C+8}{8C+7}} (Q_{\init}) ) K_i +CH^M_{\max}= 20(C+1)^2+ (2C+1)\frac{\log Q_{\init}}{\log \frac{8C+8}{8C+7}}K_i +CH^M_{\max}$.

Since $\frac{1}{8C+7} < 1$, and the taylor series for $\log(1+x)$ is an alternating series, we know that $\log(\frac{8C+8}{8C+7}) \geq \frac{1}{8C+7} - \frac{1}{2(8C+7)^2} \geq \frac{1}{16C+14}$. Therefore,
\begin{align*}
(20(C+1)^2+ (2C+1)\frac{\log Q_{\init}}{\log \frac{8C+8}{8C+7}})K_i + & CH^M_{\max} \\ 
 & \leq  (20(C+1)^2 + (2C+1)(16C+14)\log Q_{\init} )K_i + CH^M_{\max} \\ & \leq (40C^2+90C^2 \log Q_{\init})K_i + CH^M_{\max}
\end{align*}
We know that $\opt \geq M H^M_{\max}+ \max(MK_i,Q_{\init}K_i)$

It is sufficient to show that $\frac{CH^M_{\max}}{ M H^M_{\max}} \leq \frac{\epsilon}{8(C+1)}$ and that $\frac{40C^2 + 90C^2\log Q_{\init} )K_i}{\max(MK_i,Q_{\init}K_i)} \leq \frac{\epsilon}{8(C+1)} $ 

The first of these is clearly true by our choice of $M$. For the second one, we consider two cases:

\textbf{Case 1:} $Q_{\init} > M$ 
In this case, recall we chose $M$ to be sufficiently large that $(C+1)\frac{40C^2+90C^2 \log m}{m} < \frac{\epsilon}{8}$ for all $m \geq M$. Therefore, we have the desired result.

\textbf{Case 2:} $Q_{init < M}$
In this case, $\frac{(40C^2 + 90C^2\log Q_{\init} )K_i}{\max(MK_i,Q_{\init}K_i)} \leq \frac{40C^2+90C^2 \log M }{M} \leq \frac{\epsilon}{8(C+1)}$ by our choice of $M$.

Therefore, we know that the total increase in cost due to rounding for item $i$ is at most $\frac{\epsilon \opt}{8(C+1)}$
\end{proof}






Because we know there were at most $C+1$ multibatches in the initial solution, we have to apply the iterative rounding process to at most $C+1$ item types. Summing up the increase in costs due to these $C+1$ item types, by Corollary \ref{iterative_general_cost_opt_itemwise_bound} we have a proof for Theorem \ref{thm:iterative_general}
\section{Rejection Penalties}
Although we have used the language of colors, nowhere in our analysis have we assumed that a demand point has exactly one color; our algorithms work just the same if there are multiple features $c$ such that $w^c_d>0$ for a demand point $d$. This enables us to deal with rejection penalties in a straightforward manner.

The algorithm changes very slightly with rejection penalties: we first similarly solve the LP augmented with side information, generate IGOs and split the instance. However, each instance will now have an artificial $C+1$th feature. We set $w^{C+1}_d=p_d$ and the rejection limit to be $\sum_{d \in D_k} p_d r_d$ for each instance. Therefore, our integral solution will be guaranteed to have smaller penalty cost than the LP solution. The other costs can be bounded in the same way as in the previous section.

The only caveat is: The amount of side information we need will be the amount we would need if we have $C+1$ colors instead of $C$.

\section{Improving the Approximation Factor}
In this section, we will improve our approximation factors to $2.8+\epsilon$ for deadlines and $0.5(3\sqrt{5} -1)+\epsilon = 2.86 + \epsilon$ for general holding cost functions. We will first present a non-linear integer program. This non-linear integer program depends on the set $\mathcal{T}$ of IGOs. Once this non-linear program is strengthened with the appropriate amount of side information. Our two rounding algorithms will enable us to round any non-integral solution at the loss of a $1+\epsilon$ factor. The issue is that there are many feasible solutions to the original problem that will be not be feasible in the non-linear program. We will be able to construct non-integral solutions to this non-linear program that cost at most $2.8 \opt$ (in the deadlines case) and $ 2.86 \opt$ in the general holding costs case. 

In the following discussion, we will assume that the LP augmented with the appropriate amount of side information for $\cjrp$ with penalties was solved and based on its solution $(x^*,y^*,r^*)$, a set of initial general orders, $\mathcal{T}$ was chosen. We use $W_c$ to denote the total weight of all demand points in feature $c$ and $P := \sum_{d \in D} p_d r^*_d$. Using this we can construct the non-linear program. Notice that the ``we may reject weight upto $R_c$" constraint has been replaced with a ``we must service a total weight of at least $W_c - R_c$": Note that $x^{it}_{\Left},x^{it}_{\Right} \in \{0,1\}$ therefore, the non-linear constraint(s) count total weight served correctly.  We will use $s(it)$ to refer to the earliest timestep $s$ such that $x^{it*}_s>0$ and $o(it)$ to refer to the last IGO before the deadline $t$.
\begin{empheq}{align*}
\text{minimize} \quad & \sum_{s=1}^T y_s K_0 + \sum_{i=1}^N \sum_{s = 1}^T y^i_s K_i + \sum_{(i,t) \in D} \sum_{s=1}^T H_s^{it} x_s^{it} + \sum_{(i,t) \in D} p_{it}r_{it}\\
\text{subject to }
\quad & x^{it}_{\Left} \leq \sum_{s \in \mathcal{T}: s \leq o(d)} x^{it}_s \quad \text{for } (i,t) \in D\\
\quad & x^{it}_{\Right} \leq \sum_{s: o(d) < s  \leq t} x^{it}_s \quad \text{for } (i,t) \in D\\
\quad & y^i_s \leq y_s \quad \text{for all } i \in [N], s,t \in [T] \\
\quad & x_s^{it} \leq y_s^i \quad \text{for all } i \in [N], s,t \in [T], s\leq t \\
\quad & \sum_{s \in [s_1+1,s_2-1]}y_s \leq 1 \quad \text{for all pairs of consecutive IGOs } s_1,s_2 \in \mathcal{T}\\
\quad & \sum_{(i,t) \in D} (x^{it}_{\Right}+(1-x^{it}_{\Right})x^{it}_{\Left})w^c_d \geq W_c-R_c \quad \text{for all } c \in [C] \\
\quad & \sum_{(i,t) \in D} (x^{it}_{\Right}+(1-x^{it}_{\Right})x^{it}_{\Left})p_d \geq \sum_{d \in D} p_d-P \\
\quad & y_s = 1 \quad \text{for all } s \in \mathcal{T}\\
\quad & y_s^i,y_s,x^{it}_s, x^{it}_{\Right}, x^{it}_{\Left} \in \{0,1\} \quad \text{for all } (i,t) \in D, s \leq t,s \in [T]
\end{empheq}

\textbf{Fact:} Any integral solution to the above non-linear program that satisfies all constraints (except possibly ``$\sum_{s \in [s_1+1,s_2-1]}y_s \leq 1 \quad \text{for all pairs of consecutive IGOs } s_1,s_2 \in \mathcal{T}$", trivially corresponds to a feasible solution to $\cjrp$ of the same cost. 

\textbf{Note:} Unlike before, the continuous relaxation here will include the constraints $x^{it}_{\Right} \leq 1$ and $x^{it}_{\Left} \leq 1$ to ensure that $x^{it}_{\Right}+(1-x^{it}_{\Right})x^{it}_{\Left} \leq 1$ preventing a non-integral solution from ``cheating" the minimum weight served by serving the same demand more than once and counting its weight multiple times.

One can think of $x^{it}_{\Right} = 1$ as corresponding to the demand being serviced in the batch containing its deadline and $x^{it}_{\Left} = 1$ as corresponding to the demand being serviced at some IGO before its deadline.

\begin{lemma}
Given any non-integral solution of cost $\nlp$ to the above non-linear program for an instance of $\cjrp$ with $M$ chosen appropriately, we can round this solution to an integral solution of cost $(1+\epsilon) \nlp$ that satisfies all constraints with the exception of the ``$\sum_{s \in [s_1+1,s_2-1]} y_s \leq 1 \quad \text{for all pairs of consecutive IGOs } s_1,s_2 \in \mathcal{T}$"
\end{lemma}

\begin{proof}
Fix $x^{it}_{\Right}$ as constants and treat all variables corresponding to non-IGO timesteps as constants. The only variables now correspond to timesteps with IGOs and the non-trivial constraints are:
\begin{align*}
& x^{it}_{\Left} \leq \sum_{s \in \mathcal{T}: s \leq o(d)} x^{it}_s \quad \text{for } (i,t) \in D\\
& x_s^{it} \leq y_s^i \quad \text{for all } i \in [N], s,t \in \mathcal{T}, s \leq t\\
& \sum_{(i,t) \in D} (1-x^{it}_{\Right})w^c_{it} x^{it}_{\Left} \geq W_c-R_c-\sum_{(i,t) \in D} w^c_{it} x^{it}_{\Right} \quad \text{for all } c \in [C] \\
& \sum_{(i,t) \in D} (1-x^{it}_{\Right})p_{it}x^{it}_{\Left} \geq \sum_{d \in D} p_d-P - \sum_{(i,t) \in D} w^c_{it}x^{it}_{\Right}\\
& y_s^i,y_s,x^{it}_s, x^{it}_{\Left} \in \{0,1\} \quad \text{for all } (i,t) \in D, s \leq t,s \in \mathcal{T}
\end{align*}
This is equivalent to constructing a valid input for the iterative rounding algorithm where each demand point has weight $(1-x^{it}_{\Right})w^c_d$ in feature $c$ and the IGO timesteps are $\mathcal{T}$ with rejection limits $\overline{W}_c+R_c+ \sum_{d \in D} w^c_d x^d_{\Right} - W=R_c$ (where $\overline{W}_c = \sum_{d \in D} (1-x^d_{\Right})w^c_d$) and a similar adjustment for the penalty color.

This instance has a feasible LP solution $r_{it}=1-x^{it}_{\Left}$ and for all $s \in \mathcal{T}$, $\overline{y}^i_s = y^i_s, \overline{x}^{it}_s=x^{it}_s$. (remember that due to our constraints $y_s=1$ for all $s \in \mathcal{T}$). Therefore, by Theorem \ref{thm:iterative_general} we can find an integral $\overline{x},\overline{y},\overline{r}$ with total item ordering and holding cost $(1+\frac{\epsilon}{8})(\sum_{s \in \mathcal{T}} \sum_{i=1}^N K_iy^i_s + \sum_{s \in \mathcal{T}} \sum_{d \in D} H^{it}_s x^{it}_s)$.

Now, we substitute $x^{it}_s= \overline{x}^{it}_s, y^i_s = \overline{y}^i_s, x^{it}_{\Left} = \sum_{s \in \mathcal{T}, s \leq t} \overline{x}^{it}_s$.

(Note that by construction, the iterative rounding algorithm will always return a solution with $\sum_{s \in \mathcal{T}, s \leq t} \overline{x}^{it}_s \leq 1$ so $x^{it}_{\Left}$ is well defined. )

Next, we treat all variables corresponding to IGO timesteps as constants and all $x^{it}_{\Left}$ as constant. The only variables are now corresponding to non-IGO timesteps. We know that between any two successive IGOs $s_1,s_2$, $\sum_{s \in [s_1+1,s_2-1]} y_s \leq 1$. The non-trivial remaining constraints are:
\begin{align*}
& x^{it}_{\Right} \leq \sum_{s: o(d) < s  \leq t} x^{it}_s \quad \text{for } (i,t) \in D\\
& y^i_s \leq y_s \quad \text{for all } i \in [N], s,t \in [T] \\
& x_s^{it} \leq y_s^i \quad \text{for all } i \in [N], s,t \in [T], s\leq t \\
& \sum_{(i,t) \in D} (x^{it}_{\Right})w^c_d \geq W_c-R_c \quad \text{for all } c \in [C] \\
& \sum_{(i,t) \in D} (x^{it}_{\Right})p_d \geq \sum_{d \in D} p_d-P \\
& y_s^i,y_s,x^{it}_s, x^{it}_{\Right}, x^{it}_{\Left} \in \{0,1\} \quad \text{for all } (i,t) \in D, s \leq t,s \in [T]
\end{align*}
This is analogous to a valid input for the pipage rounding algorithm. The time horizon will be $\{s: s \notin \mathcal{T}, s \in [T] \}$. Notice that every demand point with $x^{it}_{\Left}=1$ is already serviced and no matter what $x^{it}_{\Right}$ is $(x^{it}_{\Right}+(1-x^{it}_{\Right})x^{it}_{\Left})=1$ and so, the set of demand points given to the pipage rounding algorithm is the set of demand points with $x^{it}_{\Left} = 0$. Each demand point will have weight $w^c_d$ and the rejection limit is $R_c$.

The starting solution will be $\overline{x} := x^{it}_s, \overline{y}_s=y_s, \overline{y}^i_s=y^i_s$ for all timesteps $s \in [T], s \notin \mathcal{T}$ and $r_{it}=1-x^{it}_{\Right}$.

We know the pipage rounding algorithm will return an integral solution $\overline{x},\overline{y},\overline{r}$ of cost $(1+\frac{\epsilon}{8})(\sum_{s \notin \mathcal{T}} K_0 y_s + \sum_{s \notin \mathcal{T}} \sum_{i=1}^N K_i y^i_s + \sum_{s \notin \mathcal{T}} \sum_{(i,t): x^{it}_{\Left}=0} H^{it}_s x^{it}_s)$

Now, we substitute $x^{it}_s = \overline{x}^{it}_s, y^i_s = \overline{y}^i_s$ and $y_s = \overline{y}_s, x^{it}_{\Right} = \sum_{ o(it)<s \leq t } x^{it}_s$.

Note that by construction, the pipage rounding algorithm will always return a solution with $\sum_{ o(it)<s \leq t } x^{it}_s \leq 1$

The penalty cost didn't increase at all because it was included as a color with a strict limit. The cost associated with opening general orders at $\mathcal{T}$ also didn't increase because of the rounding. This means we now have an integral solution which satisfies all constraints and has cost at most $(1+\frac{\epsilon}{8}) \nlp$
\end{proof}

\subsection{Deadlines: From 3 to 2.8}
The next two subsections will go over a few different ways to use the LP solution $(x^*,y^*,r^*)$ to construct a solution to the non-linear program. To determine the values of $y^i_s,x^{it}_s$ variables at IGO timesteps, we will seek to ``shift" $y^i_s$ and $x^{it}_s$ variables between two IGOs. We will often use $s_k$ to refer to an IGO and $s_{k-1}$ and $s_{k+1}$ as its' predecessor and successor IGO respectively. We will often want to look at the values of $y^i_s$ and $x^{it}_s$ on the intervals $[s_{k-1}+1,s_k]$ and $[s_k+1,s_{k+1}]$. If $s_1$ is the first and $s_L$ is the last IGO, think of $s_0=0$ and $s_{L+1}=T+1$ when defining these intervals. We have stated this purely for notational clarity; notice that $\sum_{s < s_1} y_s=0$ and $\sum_{s>s_{L}} y_s =0$ if $s_1$ is the first and $s_L$ is the last IGO.

Some of the methods in this and the next subsection will shift to both the left and the right, some will only need to scale up and shift one only to the right.

Now we show how one can construct a solution to the non-linear program of cost $2.8 \lp^{\sol}+2K_0<(2.8+\frac{\epsilon}{10}) \lp^{\sol}$

\begin{lemma}
If our holding costs are $0-\infty$, then if instead of placing IGOs when $Z_t$ reached $0,1,2.... \lfloor Z_T \rfloor, Z_T$ we placed IGOs when $Z_t$ reaches $0, \beta, 2 \beta, 3 \beta ... \lfloor \frac{Z_T}{\beta} \rfloor \beta, Z_T$ and for any IGO $s_k$ with previous IGO $s_{k-1},$ we define $\overline{y}^i_{s_k}=\frac{1}{1-\beta} \sum_{s \in [s_{k-1}+1,s_k]} y^{i*}_s$ then for $s \notin \mathcal{T}$. $y_s=y^*_s$ $y^i_s=y^{i*}_s$ and for $s \in \mathcal{T}, y_s=1, y^i_s=\overline{y}^i_s$ induces a feasible solution for the non-linear program. 
\end{lemma}

\begin{proof}
The only constraint that is not trivially satisfied is the rejection weight limit constraint. In the same way as earlier, we will define $x^{it}_s = y^i_s$ if $H^{it}_s=0$ and $0$ otherwise.
In this solution, $x^{it}_{\Left}=\min(1, \frac{1}{1-\beta} \sum_{s \leq o(it)} y^{i*}_s)$ and $x^{it}_{\Right}=\min(1, \sum_{s > o(it)} y^i_s)$.
We claim that $x^{it}_{\Right}+(1-x^{it}_{\Right})x^{it}_{\Left} \geq (1-r^{*}_{it})$. Notice that if the claim is true then the rejection limit/minimum weight serviced constraint must be satisfied.

If $x^{it}_{\Right}=1$ or $x^{it}_{\Left}$ then $x^{it}_{\Right}+(1-x^{it}_{\Right})x^{it}_{\Left} = 1 \geq 1-r^*_{it}$.

Therefore, for any demand point, 
\begin{align*}
x^{it}_{\Right}+(1-x^{it}_{\Right})x^{it}_{\Left}= & \sum_{s \leq o(it)} y^i_s+ \sum_{s > o(it)} y^i_s - (\sum_{s \leq o(it)} y^i_s)(\sum_{s > o(it)} y^i_s) \\ = & \frac{1}{1-\beta}\sum_{s \leq o(it)} y^{i*}_s + \sum_{s > o(it)} y^{i*}_s - \frac{1}{1-\beta}(\sum_{s \leq o(it)} y^{i*}_s)(\sum_{s > o(it)} y^{i*}_s) \\ \geq & \frac{1}{1-\beta}\sum_{s \leq o(it)} y^{i*}_s + \sum_{s > o(it)} y^{i*}_s -\frac{1}{1-\beta}\beta(\sum_{s \leq o(it)} y^{i*}_s)  \\
= & \sum_{s \leq o(it)} y^{i*}_s)+ \sum_{s > o(it)} y^{i*}_s \\
\geq & 1-r^*_{it}
\end{align*}

where the first inequality arises because $\sum_{s > o(it)} y^{i}_s \leq \sum_{o(it)<s \leq t} y_s \leq \beta$ (because IGOs were placed in $\beta-$sized increments of $Z_t$).

\end{proof}

Notice that the cost of the resultant solution due to this approach will be at most $(1+\frac{\epsilon}{6})((\frac{1}{\beta}+1) \LP^{\sol}_{\gen} + (\frac{1}{1-\beta}+1)\LP^{\sol}_{\itm})$ (it is $(1+\frac{\epsilon}{6})$ instead of $1+\frac{\epsilon}{8}$ because the cost of the general orders in $\mathcal{T}$ is at most $LP^{\sol}_{\gen}+2K_0$).

For similar reasons as in the monochromatic case, the algorithm described in section 5 for multiple colors gives us a solution of cost $(2+\epsilon)\LP^{\sol}_{\gen}+(3+\epsilon)\LP^{\sol}_{\itm} $

Choosing $\beta=\frac{1}{3}$ (yielding a $(4+\epsilon) \LP^{\sol}_{\gen}+(2.5+\epsilon) \LP^{\sol}_{\itm}$ and taking the best of these two algorithms gives us $2.8+\epsilon$ approximation.

\subsection{General Holding Costs: Two ways to generate a good non-integral solution}
Now we turn our attention to general holding costs. To simplify notation somewhat, we make the assumption that $\sum_{s \leq } x^{it*}_s = 1-r^*_{it}$ because otherwise we could decrease some $x^{it*}_s$ without increasing cost. We will present two slightly different methods to generate non-integral solutions, each with their own bounds on different components of the cost. Each will be parameterized by $\beta$. The optimal choice of $\beta$ leads to the desired factor.

The first approach is quite straightforward, like in the deadlines case, instead of shifting to the previous and the next IGO, we scale by $\frac{1}{1-\beta}$ and shift to the next IGO.

\begin{lemma}
If we have general holding costs, then if instead of placing IGOs when $Z_t$ reached $0,1,2.... \lfloor Z_T \rfloor, Z_T$ we placed IGOs when $Z_t$ reaches $0, \beta, 2 \beta, 3 \beta ... \lfloor \frac{Z_T}{\beta} \rfloor \beta, Z_T$ and for any IGO $s_k$ with previous IGO $s_{k-1},$ we define $\overline{y}^i_{s_k}=\frac{1}{1-\beta} \sum_{s \in [s_{k-1}+1,s_k]} y^{i*}_s$ and we similarly define $\overline{x}^{it}_{s_k} = \frac{1}{1-\beta} \sum_{s \in [s_{k-1}+1,s_k]} x^{it*}_s$ and for $s \notin \mathcal{T}$. $y_s=y^*_s$ $y^i_s=y^{i*}_s$ $x^{it}_s=x^{it*}_s$ and for $s \in \mathcal{T}, y_s=1, y^i_s=\overline{y}^i_s, x^{it}_s = \overline{x}^{it*}_s$ induces a feasible solution for the non-linear program. This solution will have cost $(\frac{1}{\beta}+1) \LP^{\sol}_{\gen} + (\frac{1}{1-\beta}+1) \LP^{\sol}_{\itm} + \frac{1}{1-\beta} \LP^{\sol}_{\hold} + 2K_0$
\end{lemma}

\begin{proof}
Once again, the only constraint which is not trivially satisfied is the rejection weight limit constraint. In the same manner as the previous lemma, we define $x^{it}_{\Right}= \min(1, \sum_{o(it)<s\leq t} x^{it}_s)$ and $x^{it}_{\Left}= \min(1,\frac{1}{1-\beta} \sum_{s \leq o(it). s \in \mathcal{T}} x^{it*}_s$. We will seek to show that $x^{it}_{\Right}+(1-x^{it}_{\Right})x^{it}_{\Left} \geq 1- r^{*}_{it}$. If either $x^{it}_{\Right}=1$ or $x^{it}_{\Left}=1$, then the desired inequality is true because the LHS is 1.

Otherwise: 

\begin{align*}
x^{it}_{\Right}+(1-x^{it}_{\Right})x^{it}_{\Left}= & \sum_{s \leq o(it)} x^{it}_s+ \sum_{s > o(it)} x^{it}_s - (\sum_{s \leq o(it)} x^{it}_s)(\sum_{s > o(it)} x^{it}_s) \\ = & \frac{1}{1-\beta}\sum_{s \leq o(it)} x^{it*}_s + \sum_{s > o(it)} x^{it*}_s - \frac{1}{1-\beta}(\sum_{s \leq o(it)} x^{it*}_s)(\sum_{s > o(it)} x^{it*}_s) \\ \geq & \frac{1}{1-\beta}\sum_{s \leq o(it)} x^{it*}_s + \sum_{s > o(it)} x^{it*}_s -\frac{1}{1-\beta}\beta (\sum_{s \leq o(it)} x^{it*}_s) \\
= & \sum_{s \leq o(it)} x^{it*}_s+ \sum_{s > o(it)} x^{it*}_s \\
\geq & 1-r^*_{it}
\end{align*}

where the first inequality arises because $\sum_{s > o(it)} x^{it}_s \leq \sum_{o(it)<s \leq t} y_s \leq \beta$ (because IGOs were placed in $\beta-$sized increments of $Z_t$).

Now we try to bound the costs: The general ordering costs and item ordering costs can bounded in the same way as in the previous subsection. To bound the holding costs, we can split the holding cost associated with any demand point into $H^{it}_{\Left}=\sum_{s \leq o(it)} H^{it}_s x^{it}_s$ and  $H^{it}_{\Right}=\sum_{s > o(it)} H^{it}_s x^{it}_s$. We can define $H^{it*}_{\Left}$ and $H^{it*}_{\Right}$ analgously using $x^*$.

Notice that for any demand point, $H^{it*}_{\Right}= H^{it}_{\Right}$. Moreover, $H^{it}_{\Left} = \sum_{s \leq o(it): s \in \mathcal{T}} H^{it}_s x^{it*}_s \leq \frac{1}{1-\beta}\sum_{s_k \leq o(it): s_k \in \mathcal{T}} H^{it}_{s_k} \sum_{s \in [s_{k-1}+1,s_k]} x^{it*}_{s} \leq \frac{1}{1-\beta} \sum_{s_k \leq o(it): s_k \in \mathcal{T}} \sum_{s \in [s_{k-1}+1,s_k]} H^{it}_s x^{it*}_{s} \leq \frac{1}{1-\beta} H^{it*}_{\Left}$

Therefore, $\sum_{(i,t) \in D} H^{it}_{\Right}+H^{it}_{\Left} \leq \frac{1}{1-\beta}\sum_{(i,t) \in D} H^{it*}_{\Right} + H^{it*}_{\Left} = \frac{1}{1-\beta} \LP^{\sol}_{\hold}$ giving us the desired result.

\end{proof}

The second approach is very similar to the one outlined in section 5. We shift item orders to both the previous and the next IGOs. However, we only shift $x$ variables forward to the next IGO. This time, we do not scale by $\frac{1}{1-\beta}$. Instead, we also shift $x$ variables backwards, but only from the batch containing the deadline of the demand point. Moreover, we shift back the necessary amount to ensure that $x^{it}_{\Right}+(1-x^{it}_{\Right})x^{it}_{\Left} \geq 1-r^*_{it}$

\begin{lemma}
If we have general holding costs, then if instead of placing IGOs when $Z_t$ reached $0,1,2.... \lfloor Z_T \rfloor, Z_T$ we placed IGOs when $Z_t$ reaches $0, \beta, 2 \beta, 3 \beta ... \lfloor \frac{Z_T}{\beta} \rfloor \beta, Z_T$ and for any IGO $s_k$ with previous IGO $s_{k-1}$ and next IGO $s_{k+1}$ we define $\overline{y}^i_{s_k}= \sum_{s \in [s_{k-1}+1,s_{k+1}]} y^{i*}_s$ and define $\overline{x}^{it}_{s_k} = \sum_{s \in [s_{k-1}+1,s_k]} x^{it*}_s$ with one exception; let $s_k=o(it)$ and then we define $\overline{x}^{it}_{o(it)} = \sum_{s \in [s_{k-1}+1,o(it)]} x^{it*}_s + \min(\sum_{s \in [o(it)+1,t]} x^{it*}_s, \frac{\beta}{1-\beta} \sum_{s \leq o(it)} x^{it*}_s)$. For $s \notin \mathcal{T}$. $y_s=y^*_s$ $y^i_s=y^{i*}_s$ $x^{it}_s=x^{it*}_s$ and for $s \in \mathcal{T}, y_s=1, y^i_s=\overline{y}^i_s, x^{it}_s = \overline{x}^{it*}_s$ induces a feasible solution for the non-linear program. This solution will have cost $(\frac{1}{\beta}+1) \LP^{\sol}_{\gen} + 3\LP^{\sol}_{\itm} + \frac{1}{1-\beta} \LP^{\sol}_{\hold} + 2K_0$
\end{lemma}

\begin{proof}
First, notice that $\min(\sum_{s \in [o(it)+1,t]} x^{it*}_s, \frac{\beta}{1-\beta} \sum_{s \leq o(it)} x^{it*}_s) \leq \sum_{s \in [o(it)+1,t]} x^{it*}_s$ and therefore $\sum_{s \leq t} x^{it}_s \leq 1$

As before, the only constraint which is not trivially satisfied is the rejection weight limit constraint. In the same manner as the previous lemma, we define $x^{it}_{\Right}= \sum_{o(it)<s\leq t} x^{it}_s$ and $x^{it}_{\Left}= \sum_{s \leq o(it). s \in \mathcal{T}} x^{it}_s$. We will seek to show that $x^{it}_{\Right}+(1-x^{it}_{\Right})x^{it}_{\Left} \geq 1- r^{*}_{it}$. 

We must consider two cases:

\textbf{Case 1:} $\sum_{s \in [o(it)+1,t]} x^{it*}_s \leq \frac{\beta}{1-\beta} \sum_{s \leq o(it)} x^{it*}_s$. In this case: $x^{it}_{\Left}= \sum_{s \leq t} x^{it*}_s=1-r_{it}^*$ which trivially implies that $x^{it}_{\Right} + (1-x^{it}_{\Right}) x^{it}_{\Left} \geq 1-r_{it}^*$.

\textbf{Case 2:} $\sum_{s \in [o(it)+1,t]} x^{it*}_s > \frac{\beta}{1-\beta} \sum_{s \leq o(it)} x^{it*}_s$. In this case, $x^{it}_{\Left} = \frac{1}{1-\beta} \sum_{s \leq o(it)} x^{it}_s$. In a manner very similarly to the previous lemma, we can conclude that:

\begin{align*}
x^{it}_{\Right}+(1-x^{it}_{\Right})x^{it}_{\Left}= & \sum_{s \leq o(it)} x^{it}_s+ \sum_{s > o(it)} x^{it}_s - (\sum_{s \leq o(it)} x^{it}_s)(\sum_{s > o(it)} x^{it}_s) \\ = & \frac{1}{1-\beta}\sum_{s \leq o(it)} x^{it*}_s + \sum_{s > o(it)} x^{it*}_s - \frac{1}{1-\beta}(\sum_{s \leq o(it)} x^{it*}_s)(\sum_{s > o(it)} x^{it*}_s) \\ \geq & \frac{1}{1-\beta}\sum_{s \leq o(it)} x^{it*}_s + \sum_{s > o(it)} x^{it*}_s -\frac{1}{1-\beta}\beta (\sum_{s \leq o(it)} x^{it*}_s) \\
= & \sum_{s \leq o(it)} x^{it*}_s+ \sum_{s > o(it)} x^{it*}_s \\
\geq & 1-r^*_{it}
\end{align*}

where the first inequality arises because $\sum_{s > o(it)} x^{it}_s \leq \sum_{o(it)<s \leq t} y_s \leq \beta$ (because IGOs were placed in $\beta-$sized increments of $Z_t$).

Therefore, the solution outlined in the lemma is feasible for the non-linear program.

Now we try to bound the costs: The general ordering costs can bounded in the same way as in the previous subsection. The item ordering costs are trivially at most thrice that in $(x^*,y^*,r^*)$ because all we did was shift extra copies of item orders between two IGOs to both the left and right.

To bound the holding costs, we can split the holding cost associated with any demand point into $H^{it}_{\Left}=\sum_{s \leq o(it)} H^{it}_s x^{it}_s$ and  $H^{it}_{\Right}=\sum_{s > o(it)} H^{it}_s x^{it}_s$. We can define $H^{it*}_{\Left}$ and $H^{it*}_{\Right}$ analgously using $x^*$.

Notice that for any demand point, $H^{it*}_{\Right}= H^{it}_{\Right}$. Moreover, $H^{it}_{\Left} = \sum_{s \leq o(it): s \in \mathcal{T}} H^{it}_s x^{it*}_s \leq \frac{1}{1-\beta}\sum_{s_k \leq o(it): s_k \in \mathcal{T}} H^{it}_{s_k} \sum_{s \in [s_{k-1}+1,s_k]} x^{it*}_{s} \leq \frac{1}{1-\beta} \sum_{s_k \leq o(it): s_k \in \mathcal{T}} \sum_{s \in [s_{k-1}+1,s_k]} H^{it}_s x^{it*}_{s} \leq \frac{1}{1-\beta} H^{it*}_{\Left}$

Therefore, $\sum_{(i,t) \in D} H^{it}_{\Right}+H^{it}_{\Left} \leq \frac{1}{1-\beta}\sum_{(i,t) \in D} H^{it*}_{\Right} + H^{it*}_{\Left} = \frac{1}{1-\beta} \LP^{\sol}_{\hold}$ giving us the desired result.

\end{proof}

We find the solution in the former of the last two lemmas with $\beta=\frac{\sqrt{\frac{f}{1-f}}}{1+\sqrt{\frac{f}{1-f}}}$ where $f=\frac{\LP^{\sol}_{\gen}}{\LP^{\sol}}$.

We then find the solution from the latter of the last two lemmas with $\beta = \frac{\sqrt{\frac{f}{1-f}}}{1+\sqrt{\frac{f}{1-f}}}$ where $f=\frac{\LP^{\sol}_{\gen}}{\LP^{\sol}_{\gen}+\LP^{\sol}_{\hold}}$

The best of these two solutions will have cost at most $\frac{1}{2}(3\sqrt{5}-1) \LP^{\sol} + 2K_0$.

\subsection{General Holding Costs: Obtaining the Improved Constant}
To simplify calculations, we will neglect the $2K_0$ term. We will use $a,b,c$ to refer to $\frac{\LP^{\sol}_{\gen}}{\LP^{\sol}},\frac{\LP^{\sol}_{\itm}}{\LP^{\sol}},\frac{\LP^{\sol}_{\hold}}{\LP^{\sol}}$. We will also call the generated solution an $(\lambda,\mu,\omega)-$approximation if it costs $\lambda \LP^{\sol}_{\gen}+ \mu \LP^{\sol}_{\itm} + \omega \LP^{\sol}_{\hold} (+2K_0)$.

The first of the two previous lemmas gives us a $(\frac{1}{\beta}+1,\frac{1}{1-\beta}+1,\frac{1}{1-\beta})-$approximate solution and the second gives us a $(\frac{1}{\beta}+1,3,\frac{1}{1-\beta})-$approximate solution

The first of these, with the optimal choice of $\beta$, $\beta=\frac{\sqrt{\frac{a}{1-a}}}{1+\sqrt{\frac{a}{1-a}}}$ leads to a $f_1(a,c):=2+2\sqrt{a(1-a)} -c$ approximation factor.

The second of these, with the optimal choice of $\beta, \beta=\frac{\sqrt{\frac{a}{c}}}{1+\sqrt{\frac{a}{c}}}$ leads to an approximation factor of $2 + 2\sqrt{ac}+b-c=3+ 2\sqrt{ac} -a-2c := f_2(a,c)$.

Notice that both of these are continuous functions of $a,c$. The best approximation factor we can obtain is $\max_{a,c} \min(f_1(a,c),f_2(a,c))$. The maximizer is either at an extreme point of the region of possibilities ($a,c \geq 0, a+c \leq 1$) or at a point where the gradient of the smaller function is $0$ or at some $a,c$ where $f_1(a,c)=f_2(a,c)$. 

Some calculation reveals that the maximizer is when $f_1(a,c)=f_2(a,c)$. This implies that $2\sqrt{a(1-a)}-2\sqrt{ac}=1-a-c$. Solving for $c$, besides the degenerate solutions $a=1,c=0$ and $c=1,a=0$ (both of which yield an approximation factor of $2$ or less), we get another solution; when $0.2 \leq a \leq 0.5)$, $c=3a-4\sqrt{a(1-a)}+1$. Substituting this in and finding the maximum value gives us $\frac{1}{2}(3\sqrt{5}-1)$

\bibliographystyle{splncs04}
\bibliography{base}

\begin{thebibliography}{10}
\providecommand{\url}[1]{\texttt{#1}}
\providecommand{\urlprefix}{URL }
\providecommand{\doi}[1]{https://doi.org/#1}

\bibitem{AgeevS04}
Ageev, A.A., Sviridenko, M.: Pipage rounding: {A} new method of constructing
  algorithms with proven performance guarantee. J. Comb. Optim.  \textbf{8}(3),
   307--328 (2004)

\bibitem{AneggAKZ20}
Anegg, G., Angelidakis, H., Kurpisz, A., Zenklusen, R.: A technique for
  obtaining true approximations for k-center with covering constraints. In:
  Bienstock, D., Zambelli, G. (eds.) Integer Programming and Combinatorial
  Optimization - 21st International Conference, {IPCO} 2020, London, UK, June
  8-10, 2020, Proceedings. Lecture Notes in Computer Science, vol. 12125, pp.
  52--65. Springer (2020)

\bibitem{Bandyapadhyay0P19}
Bandyapadhyay, S., Inamdar, T., Pai, S., Varadarajan, K.R.: A constant
  approximation for colorful k-center. In: Bender, M.A., Svensson, O., Herman,
  G. (eds.) 27th Annual European Symposium on Algorithms, {ESA} 2019, September
  9-11, 2019, Munich/Garching, Germany. LIPIcs, vol.~144, pp. 12:1--12:14.
  Schloss Dagstuhl - Leibniz-Zentrum f{\"{u}}r Informatik (2019)

\bibitem{BienkowskiBCDNSSY15}
Bienkowski, M., Byrka, J., Chrobak, M., Dobbs, N., Nowicki, T., Sviridenko, M.,
  {\'S}wirszcz, G., Young, N.E.: Approximation algorithms for the joint
  replenishment problem with deadlines. Journal of Scheduling  \textbf{18}(6),
  545--560 (2015)

\bibitem{ChakrabartyGK20}
Chakrabarty, D., Goyal, P., Krishnaswamy, R.: The non-uniform \emph{k}-center
  problem. {ACM} Trans. Algorithms  \textbf{16}(4),  46:1--46:19 (2020)

\bibitem{CharikarKMN01}
Charikar, M., Khuller, S., Mount, D.M., Narasimhan, G.: Algorithms for facility
  location problems with outliers. In: Kosaraju, S.R. (ed.) Proceedings of the
  Twelfth Annual Symposium on Discrete Algorithms, January 7-9, 2001,
  Washington, DC, {USA}. pp. 642--651. {ACM/SIAM} (2001)

\bibitem{ChekuriQZ2019}
Chekuri, C., Quanrud, K., Zhang, Z.: On approximating partial set cover and
  generalizations. CoRR  \textbf{abs/1907.04413} (2019)

\bibitem{Chen08}
Chen, K.: A constant factor approximation algorithm for \emph{k}-median
  clustering with outliers. In: Teng, S. (ed.) Proceedings of the Nineteenth
  Annual {ACM-SIAM} Symposium on Discrete Algorithms, {SODA} 2008, San
  Francisco, California, USA, January 20-22, 2008. pp. 826--835. {SIAM} (2008)

\bibitem{HarrisPST19}
Harris, D.G., Pensyl, T.W., Srinivasan, A., Trinh, K.: A lottery model for
  center-type problems with outliers. {ACM} Trans. Algorithms  \textbf{15}(3),
  36:1--36:25 (2019)

\bibitem{InamdarV18}
Inamdar, T., Varadarajan, K.R.: On the partition set cover problem. CoRR
  \textbf{abs/1809.06506} (2018)

\bibitem{JiaSS20}
Jia, X., Sheth, K., Svensson, O.: Fair colorful k-center clustering. In:
  Bienstock, D., Zambelli, G. (eds.) Integer Programming and Combinatorial
  Optimization - 21st International Conference, {IPCO} 2020, London, UK, June
  8-10, 2020, Proceedings. Lecture Notes in Computer Science, vol. 12125, pp.
  209--222. Springer (2020)

\bibitem{KrishnaswamyLS18}
Krishnaswamy, R., Li, S., Sandeep, S.: Constant approximation for k-median and
  k-means with outliers via iterative rounding. In: Diakonikolas, I., Kempe,
  D., Henzinger, M. (eds.) Proceedings of the 50th Annual {ACM} {SIGACT}
  Symposium on Theory of Computing, {STOC} 2018, Los Angeles, CA, USA, June
  25-29, 2018. pp. 646--659. {ACM} (2018)

\bibitem{LeviRS06}
Levi, R., Roundy, R., Shmoys, D.B.: Primal-dual algorithms for deterministic
  inventory problems. Math. Oper. Res.  \textbf{31}(2),  267--284 (2006)

\bibitem{NonnerS09}
Nonner, T., Souza, A.: A 5/3-approximation algorithm for joint replenishment
  with deadlines. In: Du, D., Hu, X., Pardalos, P.M. (eds.) Combinatorial
  Optimization and Applications, Third International Conference, {COCOA} 2009,
  Huangshan, China, June 10-12, 2009. Proceedings. Lecture Notes in Computer
  Science, vol.~5573, pp. 24--35. Springer (2009)

\bibitem{WagnerW58}
Wagner, H.M., Whitin, T.M.: Dynamic version of the economic lot size model.
  Management Science  \textbf{5}(1) (1958)

\end{thebibliography}
\end{document}